\documentclass{article}
\usepackage[utf8]{inputenc}
\usepackage{xcolor, graphicx, algorithm, mathtools, amsthm, grffile, wasysym}
\usepackage[english]{babel}

\usepackage{blindtext}
\usepackage{amsmath,amssymb}
\usepackage{comment}
\usepackage{subcaption}
\usepackage{colortbl}

\usepackage[margin=1in]{geometry}

\let\realbfseries=\bfseries
\def\bfseries{\realbfseries\boldmath}

\newtheorem{theorem}{Theorem}[section]
\newtheorem{corollary}[theorem]{Corollary}
\newtheorem{lemma}[theorem]{Lemma}

\def\defn#1{\textbf{\textit{#1}}}
\def\degree{\operatorname{degree}}

\def\weirdSAT{N3P-3SAT-2P-E1N}
\def\noEweirdSAT{N3P-3SAT-2P-1N}
\def\weakweirdSAT{N3P-3SAT-3-1N}
\def\pospathSAT{literal-matching \weirdSAT}
\def\PospathSAT{Literal-matching \weirdSAT}

\definecolor{tile}{gray}{0.8}
\definecolor{tilelabel}{gray}{0.9}
\def\growbox#1{%
  \setbox1=\hbox{#1}%
  \ifdim \wd1 < 5pt
    \hbox to 5pt{\hfil\box1\hfil}%
  \else
    \box1%
  \fi}
\def\shrink#1{\hbox to 0pt{\hss$#1$\hss}}
\def\TILE#1#2#3#4{\ensuremath{\raisebox{0.2ex}{$\scriptstyle #1$} {\overset{\smash[b]{\scriptstyle #2}}{\underset{\smash[t]{\scriptstyle #4}}{\raisebox{0ex}{\scalebox{1.0}{$\square$}}}}} \raisebox{0.2ex}{$\scriptstyle #3$}}}
\def\TRITILEEQ#1#2#3{\ensuremath{\raisebox{1ex}{$\scriptstyle #1\!$} {\underset{\smash[t]{\scriptstyle #3}}{\raisebox{-0.5ex}{\scalebox{1.2}{$\triangle$}}}} \raisebox{1ex}{$\!\scriptstyle #2$}}}
\def\TRITILER#1#2#3{\ensuremath{\raisebox{1ex}{$\scriptstyle #1\!\!$} {\underset{\smash[t]{\scriptstyle #3}}{\raisebox{-0.5ex}{\scalebox{1.2}{$\angle$\hspace{-1.5pt}\reflectbox{$\angle$}}}} \raisebox{1ex}{$\!\!\scriptstyle #2$}}}}

\theoremstyle{remark}

\theoremstyle{definition}
\newtheorem{definition}{Definition}[section]

\usepackage[unicode]{hyperref} %
\hypersetup{breaklinks, bookmarksnumbered, bookmarksopen, bookmarksopenlevel=2, hypertexnames, pageanchor}
{\makeatletter \hypersetup{pdftitle={\@title}}}
\title{Edge Matching with Inequalities, Triangles, \\ Unknown Shape, and Two Players}
\author{%
  Jeffrey Bosboom%
    \thanks{MIT Computer Science and Artificial Intelligence Laboratory,
      32 Vassar St., Cambridge, MA 02139, USA,
      \protect\url{jbosboom@csail.mit.edu},\protect\url{{czchen,ikdc,scompton,mcoulomb,edemaine,mdemaine,ivanaf,dylanhen,achester,ineq,whu2704,ezzluo,lillianz}@mit.edu},
      \protect\url{charlotte_z_chen@yahoo.com}}
\and
  Charlotte Chen%
    \footnotemark[1]
\and
  Lily Chung%
    \footnotemark[1]
\and
  Spencer Compton%
    \footnotemark[1]
\and
  Michael Coulombe%
    \footnotemark[1]
\and
  Erik~D.~Demaine%
    \footnotemark[1]
\and
  Martin L. Demaine%
    \footnotemark[1]
\and
  Ivan Tadeu Ferreira Antunes Filho%
    \footnotemark[1]
\and
  Dylan~Hendrickson%
    \footnotemark[1]
\and
  Adam Hesterberg%
    \footnotemark[1]
\and
  Calvin Hsu%
    \footnotemark[1]
\and
  William Hu%
    \footnotemark[1]
\and
  Oliver Korten%
    \smash{\thanks{Department of Computer Science,
      Tufts University, Medford, MA, USA,
      \protect\url{oliverkorten123@gmail.com}}}
\and
  Zhezheng Luo%
    \footnotemark[1]
\and
  Lillian Zhang%
    \footnotemark[1]
}
\date{}

{\makeatletter
 \gdef\xxxmark{%
   \expandafter\ifx\csname @mpargs\endcsname\relax %
     \expandafter\ifx\csname @captype\endcsname\relax %
       \marginpar{xxx}%
     \else
       xxx %
     \fi
   \else
     xxx %
   \fi}
 \gdef\xxx{\@ifnextchar[\xxx@lab\xxx@nolab}
 \long\gdef\xxx@lab[#1]#2{\textbf{[\xxxmark #2 ---{\sc #1}]}}
 \long\gdef\xxx@nolab#1{\textbf{[\xxxmark #1]}}
 \long\gdef\xxx@lab[#1]#2{}\long\gdef\xxx@nolab#1{}%
}

\begin{document}

\maketitle

\begin{abstract}
  We analyze the computational complexity of several new variants of
  edge-matching puzzles.
  First we analyze inequality (instead of equality) constraints
  between adjacent tiles,
  proving the problem NP-complete for strict inequalities
  but polynomial-time solvable for nonstrict inequalities.
  Second we analyze three types of triangular edge matching,
  of which one is polynomial-time solvable and the other two are NP-complete;
  all three are \#P-complete.
  Third we analyze the case where no target shape is specified, and we merely
  want to place the (square) tiles so that edges match (exactly);
  this problem is NP-complete.
  Fourth we consider four 2-player games based on $1 \times n$ edge matching,
  all four of which are PSPACE-complete.
  Most of our NP-hardness reductions are parsimonious, newly proving \#P and
  ASP-completeness for, e.g.,\ $1 \times n$ edge matching.
  Along the way, we prove \#P- and ASP-completeness of
  planar 3-regular directed Hamiltonicity; we give linear-time algorithms
  to find antidirected and forbidden-transition Eulerian paths;
  and we characterize the complexity of new partizan variants of the
  Geography game on graphs.
\end{abstract}

\section{Introduction}

In an \defn{edge-matching puzzle}, we are given several tiles
(usually identical in shape), where each tile has a label on each edge,
and the goal is to place all the tiles (usually via translation and rotation)
into a given shape such that shared edges between adjacent tiles have
compatible labels.
In \defn{unsigned} edge matching, labels are compatible if they are identical ($a$ matches $a$ and nothing else); in \defn{signed} edge matching, labels have signs (e.g., $+a$ and $-a$), and two labels are compatible if they are negations of each other ($+a$ matches $-a$ and nothing else, and vice versa).
Physical edge-matching puzzles date back to the 1890s \cite{Thurston-patent};
perhaps the most famous example is \emph{Eternity II} which offered a US\$2,000,000 prize for a solution before 2011 \cite{eternity2-wiki}.

\subsection{Previous Work}
The complexity of edge-matching puzzles has been studied since 1966
\cite{Berger-1966}.
The most relevant work to this paper is from two past JCDCG conferences.
In 2007, Demaine and Demaine \cite{demaine2007jigsaw} proved that signed and unsigned edge-matching square-tile puzzles are NP-complete and equivalent to both jigsaw puzzles and polyomino packing puzzles.
In 2016, Bosboom et al.~\cite{bosboom2017even} proved that signed and unsigned edge-matching square-tile puzzles are NP-complete even when the target shape is a $1 \times n$ rectangle, and furthermore hard to approximate within some constant factor.
Our work on $1 \times n$ triangle edge-matching puzzles is inspired by an open problem proposed in the latter paper.

\definecolor{header}{rgb}{0.29,0,0.51}
\definecolor{gray}{rgb}{0.85,0.85,0.85}
\def\header#1{\textcolor{white}{\textbf{#1}}}
\begin{table}[h]
\centering
\small
\setlength\tabcolsep{0.9\tabcolsep}
\begin{tabular}{lclll}%
\rowcolor{header}
\header{Compatibility} & \header{Board} & \header{Tiles} & \header{Players} & \header{Complexity} \\
$<$ & $1 \times n$ & square & 1-player & NP-complete\\
$\leq$ & $m \times n$ & square & 1-player & P \\
\rowcolor{gray}
Signed/unsigned & $1 \times n$ & square & 1-player & NP/\#P-complete, (2-)ASP-hard* \\
Signed/unsigned & $1 \times n$ & equilateral triangle & 1-player & NP/\#P-complete, (2-)ASP-hard* \\
Signed/unsigned & $1\times n$ & right triangle (hypotenuse contact) & 1-player & NP/\#P-complete, (2-)ASP-hard* \\
Signed/unsigned & ${\sqrt{2} \over 2} \times n$ & right triangle (leg contact) & 1-player & $\in{}$P, \#P-complete \\
\rowcolor{gray}
Signed/unsigned & $O(1) \times n$ & square/triangular with $O(1)$ colors & 1-player & $\in{}$P\\
Signed/unsigned & shapeless & square & 1-player & NP/\#P/ASP-complete \\
\rowcolor{gray}
Signed/unsigned & $1 \times n$ & square & \llap{impartial }2-player & PSPACE-complete\\
\rowcolor{gray}
Signed/unsigned & $1 \times n$ & square & \llap{partizan }2-player & PSPACE-complete\\
\end{tabular}
\vspace*{-1ex}
\caption{Our results on edge-matching puzzles.  *Our proof gives ASP-completeness for $1 \times n$ edge matching only when at least one boundary edge is colored; otherwise, each solution can be rotated 180 degrees to form another valid solution, so we get 2-ASP-hardness (NP-hard to find a third solution given two).}
\label{tab:results}
\end{table}

\subsection{Our Results: Edge Matching}

Table~\ref{tab:results} summarizes our results in edge matching,
described in more detail below.

\paragraph{Inequality edge matching.}
Our most involved result is an NP-hardness proof for a new ``$<$'' compatibility condition, where edge labels are numbers, horizontally adjacent edges match if the left edge's number is less than the right edge's number, and vertically adjacent edges match if the top edge's number is less than the bottom edge's number.
Figure~\ref{2x3 less than} shows an example.
In Section~\ref{sec:ineq}, we prove NP-hardness of $<$-compatible $1 \times n$ edge matching by reduction from another new NP-hard problem, Interval-Pair Cover.
The $\leq$-compatibility condition (where equal numbers also match, or we assume all labels are distinct) turns out to be substantially easier: even rectangular puzzles turn out to be always solvable, and we give a polynomial-time algorithm.%

\begin{figure}
  \centering
  \includegraphics[scale=2]{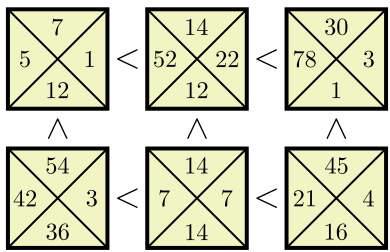}
  \caption{A solved $2 \times 3$ \hbox{$<$-compatible} edge-matching puzzle.
    This solution is valid because $1 < 52$ and $22 < 78$ in the top row,
    $3 < 7$ and $7 < 21$ in the bottom row, and $12 < 54$, $12 < 14$, and
    $1 < 45$ in the columns.}
  \label{2x3 less than}
\end{figure}

\paragraph{ASP/\#P-completeness for $1 \times n$ edge matching.}
In Section~\ref{sec:ASP},
we analyze edge matching for the first time from the perspective of the number
of solutions to an instance,
which is relevant to constructing puzzles with unique solutions.
Specifically, we prove ASP-completeness for signed and unsigned $1 \times n$ edge-matching puzzles when the left boundary edge is colored (to prevent trivial $180^\circ$ rotation of solutions), and 2-ASP-hardness and \#P-completeness even if the boundary is colorless.

Recall the following definitions of FNP, ASP-complete, and \#P-complete.
\defn{FNP} is a variant of NP that actually specifies the valid
certificates/solutions for an instance
(instead of just requiring that they exist); that is,
an FNP problem is a relation between instances and polynomial-length
certificates/solutions that can be checked in polynomial time.
For edge matching problems, the certificate we consider is a valid placement
of the given tiles within the given shape.
An FNP problem $\Pi$ is \defn{ASP-complete} \cite{Yato-Seta-2003}
if every FNP problem has a polynomial-time \defn{parsimonious reduction}
(preserving the number of solutions) to $\Pi$ along with a polynomial-time
bijection between solutions of the two problems.
ASP-completeness implies that the \defn{$k$-ASP} version of the FNP problem ---
given an instance and $k$ solutions to it, decide whether there is another
solution --- is NP-hard \cite{Yato-Seta-2003}.
An FNP problem is \defn{\#P-complete} \cite{Valiant-1979-perm}
if counting the number of solutions
is as hard as counting the number of solutions to any FNP problem,
which is implied by a reduction that is \defn{$c$-monious}, i.e.,
that multiplies the number of solutions by a computable factor $c \geq 1$.%
\footnote{This terminology naturally generalizes ``parsimonious''
  ($c=1$), and was introduced in an MIT class in 2014 \cite{6.892-L10}.}
Our reductions to $1 \times n$ edge matching are the first to be parsimonious
or, when global $180^\circ$ rotation is allowed, 2-monious.

\paragraph{Triangular edge matching.}
The conclusion of \cite{bosboom2017even} claimed that the paper's results extended to equilateral-triangle edge matching, but the proposed simulation of squares by triangles is incorrect because it constrains the orientation of the simulated squares.
In Section~\ref{equilateral-triangle-hardness}, we extend our $1 \times n$ parsimonious proof to obtain NP/\#P/ASP-completeness for signed and unsigned edge matching with equilateral triangles, with or without reflection.

For right isosceles triangles, there are two natural ``$1 \times n$'' arrangements.
For clarity, we assume the legs of the triangles have length 1.
If we still want a height-$1$ tiling, then length-$\sqrt 2$ hypotenuses are forced to match, so matching is NP/\#P/ASP-complete by simulation of squares.
But if we ask for a height-$\sqrt 2 \over 2$ tiling, so only legs match, we show in Section~\ref{right triangles} that, surprisingly, both signed and unsigned edge matching can be solved in polynomial time using an algorithm based on Eulerian paths.
Nonetheless, the latter problems are still \#P-complete.

\paragraph{Shapeless edge matching.}
In Section~\ref{sec:shapeless}, we prove that square-tile edge-matching puzzles remain NP/\#P/\allowbreak ASP-complete when the goal is to connect all tiles in any (unspecified) single connected shape, with either signed or unsigned compatibility.
For \#P- and ASP-completeness, we need to give some tile a fixed position in the plane (translation and rotation) to make the number of solutions finite.
The proof builds a unique spiral frame that effectively forces a $1 \times n$ edge-matching puzzle with a fixed left boundary color.

\paragraph{2-player edge matching.}
In Section~\ref{sec:2-player}, we consider natural 2-player variants of $1 \times n$ edge-matching puzzles, where the left boundary edge of the rectangle has a prespecified color, players alternate placing a tile in the leftmost empty cell that matches the edge color to the left, and the first player unable to move loses (\defn{normal play}).
We prove PSPACE-completeness for four variants of this problem: both signed and unsigned square-tile edge matching, and both when players can play any remaining tile from a shared pool (\defn{impartial}) and when players play from separate pools of tiles (\defn{partizan}).

\subsection{Our Results: Not Edge Matching}

Along the way to proving our results on edge matching, we derive other
results of possible independent interest in graph algorithms/complexity.

\paragraph{Hamiltonicity parsimony.}

In Section~\ref{sec:Ham ASP}, we prove \#P- and ASP-completeness of the Hamiltonian cycle problem in planar 3-regular \emph{directed} graphs, by modifying the clause gadget of Plesn\'ik's \emph{NP}-hardness proof \cite{plesnik} and parsimoniously reducing from 1-in-3SAT instead of 3SAT.
Previous work showed the analogous \emph{undirected} problem ASP-complete (and \#P-complete) in planar graphs of maximum degree 3 \cite{Seta02thecomplexities}.
We also prove \#P- and ASP-completeness of the Hamiltonian path problem with specified start and end vertices in planar 3-regular directed graphs.

\paragraph{Antidirected Eulerian paths.}

In Section~\ref{sec:alternating euler}, we characterize when a directed graph
admits an \defn{antidirected Eulerian path}
\cite{Berman-1978,Fleischner-1990,Zitnik-1996},
that is, a path%
\footnote{Throughout this paper, we follow the half-standard terminology
  that paths and cycles are allowed to repeat vertices and/or edges
  (though we will rarely allow repeated edges).
  In a different half-standard terminology, these notions are called
  ``walks/trails'' and ``circuits''.
  If a path or cycle makes no such repetitions, it is called \defn{simple}.}
that alternates between going forward and going backward along
directed edges and visits every edge (in either direction) exactly once.
(Such directed graphs are called \defn{aneulerian}
\cite{Berman-1978,Fleischner-1990,Zitnik-1996}.)
Specifically, we show how to reduce this problem to finding an Eulerian path
in a modified graph, enabling solution in linear time.
Although antidirected Eulerian paths were introduced over 50 years ago
\cite{Berman-1978},
their existence does not seem to have been characterized before our work
and a recent independent discovery \cite{alternating-euler-independent}.

\paragraph{Forbidden-transition Eulerian paths.}
In Section~\ref{sec:forbidden euler}, we give linear-time algorithms
to find Eulerian paths or antidirected Eulerian paths
when certain monochromatic edge-to-edge transitions are forbidden,
extending past work by Kotzig \cite{Kotzig1968} to be algorithmic
(and to the antidirected case).
Specifically, each vertex can define a partition of its incident edges
into groups, and the problem forbids the Eulerian path from passing through
the vertex via two edges from the same group.

\paragraph{Partizan Geography game.}

We introduce eight new \emph{partizan} variants of Geography where the two
players have different available moves, and characterize their complexity.
Specifically, in vertex-partizan geography, vertices have two different colors,
and each player can only move to vertices of their color; while in
edge-partizan geography, edges have two different colors, and each player
can only move along edges of their color.
We can consider either variant for both Vertex and Edge Geography
(where vertices and edges, respectively, cannot be repeated by either player),
and in directed or undirected graphs, resulting in eight possible variants.
Table~\ref{tab:partizan-geography} summarizes our results
from Section~\ref{sec:par},
which prove every variant either polynomial or PSPACE-complete.

\begin{table}
\centering
\begin{tabular}{lrll}
  \rowcolor{header}
  \header{Graph} & \header{Partizan} & \header{Geography} & \header{Complexity} \\
  \hline 
  undirected & vertex   & vertex           & polynomial \\
\rowcolor{gray}
  undirected & vertex   & edge             & polynomial \\
  undirected & edge     & vertex           & PSPACE-complete \\
\rowcolor{gray}
  undirected & edge     & edge             & PSPACE-complete \\
  directed   & vertex   & vertex           & PSPACE-complete \\
\rowcolor{gray}
  directed   & vertex   & edge             & PSPACE-complete \\
  directed   & edge     & vertex           & PSPACE-complete \\
\rowcolor{gray}
  directed   & edge     & edge             & PSPACE-complete
\end{tabular}
\caption{Partizan geography results}
\label{tab:partizan-geography}
\end{table}

\section{Edge Matching with Inequalities}
\label{sec:ineq}

In this section, we analyze the complexity of the following problems:

\begin{definition}
\defn{$m \times n$ $<$-compatible edge matching} is the following problem:

\textbf{Instance:} $m n$ unit-square tiles, where each tile is defined by four numbers,
one for each side. We use \TILE{a}{b}{c}{d} to represent a unit-square tile with numbers $a,b,c,d$.

\textbf{Question:} Can the $m n$ tiles cover an $m \times n$ rectangle such that
\begin{itemize}
\item for every two horizontally adjacent tiles, the left tile's right number is strictly smaller than the right tile's left number; and
\item for every two vertically adjacent tiles, the top tile's bottom number is strictly smaller than the bottom tile's top number?
\end{itemize}

The related problem \defn{$\leq$-compatible edge matching} is defined similarly, except that we do not require strict inequalities among the numbers. 
\end{definition}

\subsection{Polynomial-Time Algorithm for $\leq$-Compatible Edge Matching}

\begin{theorem}
$m \times n$ $\leq$-compatible edge-matching puzzles are always solvable and a solution can be found in $O(mn\log(mn))$ time.
\end{theorem}

\begin{proof}
Rotate each tile \TILE{A}{B}{C}{D} such that $A \ge C$ and $B \ge D$.  Then sort the tiles in ascending order by $D$ and place them in the board in row-major order.  Because $B \ge D$, sorting by $D$ ensures all tiles are vertically $\leq$-compatible.  Then sort the tiles in each row in ascending order by $C$.  Because $A \ge C$, sorting by $C$ ensures all tiles in the row are horizontally $\leq$-compatible.  Being both vertically and horizontally $\leq$-compatible, this is a compatible tiling.  This algorithm runs in $O(mn\log(mn))$ time from the sorting.
\end{proof}

The following special cases of the $m \times n$ $<$-compatible edge-matching
puzzles are tractable:

\begin{corollary}
$m \times n$ $<$-compatible edge-matching puzzles in which all edge labels are distinct are always solvable and a solution can be found in polynomial time.
\end{corollary}

\begin{theorem}
$1 \times n$ $<$-compatible edge-matching puzzles in which every tile has at least one pair of parallel sides with unequal labels are always solvable and a solution can be found in polynomial time. 
\end{theorem}

\begin{proof}
Rotate each tile \TILE{A}{B}{C}{D} such that $A > C$.  If there are two pairs of unequal parallel sides, then choose arbitrarily.  Now sort all tiles in ascending order by $A$, breaking ties arbitrarily, and place them in the board in row-major order.  Let $A_i$ and $C_i$ be the left and right numbers of tile $i$.  From sorting, we know that $A_i \leq A_{i+1}$, and from our rotation of the tiles, we know that $C_i < A_i$.  Composing the inequalities gives $C_i < A_{i+1}$, which is the $<$-compatibility condition, so this is a compatible tiling.
\end{proof}

\subsection{NP-hardness of $1 \times n$ $<$-Compatible Edge Matching}

To show NP-hardness of $<$-compatible edge matching, we start from the known
NP-hard problem {\weirdSAT} \cite{ding2011minimum} defined in
Section~\ref{sec:weirdSAT}. In Section
\ref{FirstNPReduction}, we reduce {\weirdSAT} to a novel variant
{\pospathSAT}. In Section \ref{SecondNPReduction}, we reduce {\pospathSAT} to
a new problem, Interval-Pair Cover, which implies NP-hardness of $1 \times n$
$<$-compatible edge matching.

\subsubsection{\weirdSAT}
\label{sec:weirdSAT}

Our starting point is the following variant of SAT
(named to follow notation from \cite{ivan-thesis}):

\begin{definition}
An instance of \defn{\weirdSAT} is an instance of 3SAT, consisting of
$n$ variables $x_1, x_2, \dots, x_n$ and
$m$ clauses each with at most three literals, where each literal is of the form
$x_i$ (positive) or $\neg x_i$ (negative),
satisfying the following constraints:
\begin{enumerate}
\item \defn{N3P}: Every clause has at least one negative literal
      (i.e., no clause has three positive literals).
\item \defn{2P}: Every variable $x_i$ appears in at most two positive literals $x_i$.
\item \defn{E1N}: Every variable $x_i$ appears in exactly one negative literal $\neg x_i$.
\end{enumerate}
\end{definition}

Ding et al.~\cite{ding2011minimum} proved that {\weirdSAT} is NP-complete.
In fact, they proved NP-completeness of a slightly more general problem,
{\weakweirdSAT}, which constrains each variable to appear in at most three
literals, at most one of which is negative.
But any variable with zero negative occurrences can be eliminated
(setting it to true), so by repeated application of this process,
we attain the E1N property.
Because each variable appears in at most three literals, at most two of them
are positive, so we also have the 2P property.
Thus we reduce {\weakweirdSAT} to {\weirdSAT}.

\subsubsection{Reduction from {\weirdSAT} to {\pospathSAT}} \label{FirstNPReduction}

Define the \defn{shared-literal graph} of a
3SAT instance to have one vertex for each clause,
and connect two clauses by an edge for each literal they share;
see Figure~\ref{shared literal graph}.
For a {\weirdSAT} instance, the shared-literal graph has two additional properties.
By the E1N constraint, every edge corresponds to a
shared \emph{positive} literal.
By the 2P property, the shared-literal graph has maximum degree~$2$.
We will show that we can in fact reduce the shared-literal graph to maximum degree~$1$.

\begin{figure}
  \centering
  \subcaptionbox{3SAT formula}{\includegraphics[page=1,width=0.45\linewidth,trim=0 45 0 0,clip]{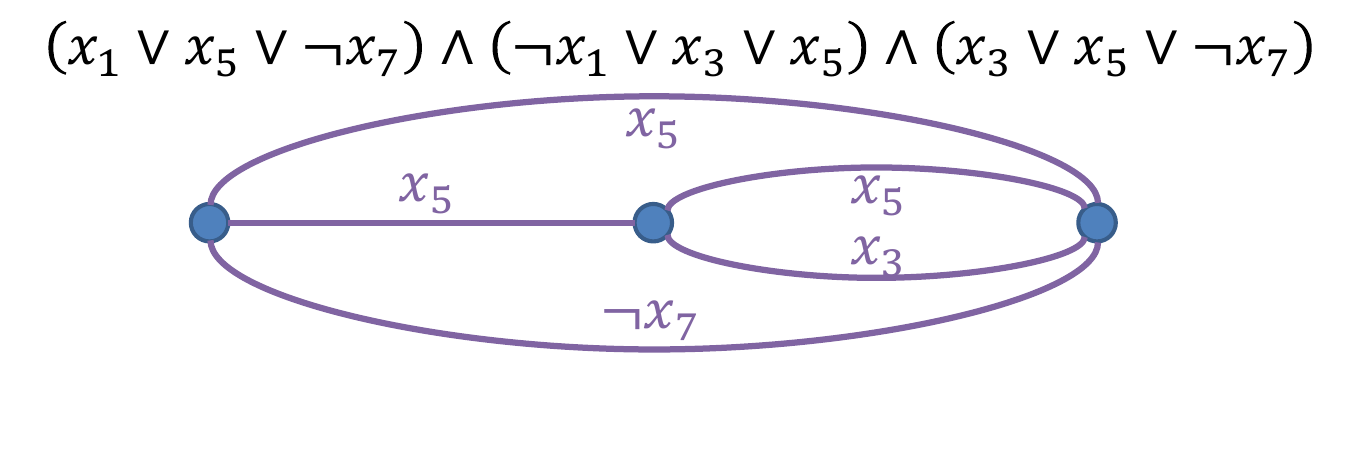}}\hfil\hfil
  \subcaptionbox{\label{shared literal graph weird}{\noEweirdSAT} formula}{\includegraphics[page=2,width=0.45\linewidth,trim=0 45 0 0,clip]{ppt_figs}}
  \caption{Shared-literal graph: two examples.}
  \label{shared literal graph}
\end{figure}

\begin{definition}
A \defn{\pospathSAT} instance is an instance of {\weirdSAT} whose shared-literal graph
is a matching.
\end{definition}

\begin{theorem} \label{pospathSAT thm}
  {\PospathSAT} is NP-complete.
\end{theorem}

\begin{proof}
Trivially, {\pospathSAT} $\in$ NP.
We reduce {\weirdSAT} to {\pospathSAT} to show {\pospathSAT} is NP-hard. 
Refer to Figure~\ref{pospathSAT}.

\begin{figure}
  \centering
  \subcaptionbox{Oriented {\weirdSAT} instance from Figure~\ref{shared literal graph weird}}{\includegraphics[page=3,width=0.45\linewidth]{ppt_figs}}\hfil\hfil
  \subcaptionbox{Reduced {\pospathSAT} instance}{\includegraphics[page=4,width=0.45\linewidth]{ppt_figs}}
  \caption{Reduction from {\weirdSAT} to {\pospathSAT} of Theorem~\ref{pospathSAT thm}.}
  \label{pospathSAT}
\end{figure}

First we orient the shared-literal graph to have maximum indegree and maximum
outdegree~$1$.
Because the shared-literal graph is maximum degree~$2$,
every connected component is either a path or a cycle.
Direct each path from one end to the other, and direct each cycle cyclically.

Reduction:
For each edge $(c,d)$ in the directed shared-literal graph,
corresponding to a shared literal $x_i$,
replace the occurrence of $x_i$ in $d$ with a new helper variable~$h_i$.
Additionally, create a new helper clause $\neg h_i \lor x_i$,
i.e., $h_i \Rightarrow x_i$.

This reduction conserves occurrences of the original (nonhelper) variables, and each helper variable appears positively once (replacing some $x_i$ in an original clause) and negatively once (in the helper clause), so the transformed formula is still {\weirdSAT}.

The transformed formula is satisfiable under an augmented truth assignment $\sigma_{X, H} = \sigma_X \cup \sigma_H$ if and only if the original formula is satisfiable under $\sigma_X$.  If $h_i$ satisfies an original clause (by being true), the helper clause enforces that $x_i$ is also true.  If $x_i$ is false, the helper clause enforces that $h_i$ is also false, and so cannot satisfy the original clause it is a member of.  Thus if $\sigma_{X,H}$ satisfies the transformed formula, $\sigma_X$ satisfies the original formula.  Variable $h_i$ can be false when $x_i$ is true, but as $x_i$ already satisfies $h_i$'s helper clause and $h_i$ always appears positively in its original clause, such an assignment cannot satisfy more clauses than if $h_i$ were true.  Thus if $\sigma_X$ satisfies the original formula, $\sigma_{X,H} = \sigma_X \cup \{h_i = \sigma_X(x_i)\}$ satisfies the transformed formula.

After replacing the occurrence of $x_i$ in clause $d$, each edge $(c,d)$ in the original formula's directed shared-literal graph corresponds to an edge between $c$ and the helper clause containing $x_i$ in the transformed formula's shared-literal graph, so original clauses have degree at most 1.  Each helper variable appears only once in each polarity, so helper variables do not give rise to edges in the shared-literal graph.  Thus all helper clauses have degree 1.  Then the transformed formula's shared-literal graph has maximum degree~1, and so is a matching.
\end{proof}

\subsubsection{Reduction from {\pospathSAT} to Interval-Pair Cover} \label{SecondNPReduction}
To begin, we define a new problem Interval-Pair Cover;
refer to Figure~\ref{interval pair cover}.

\begin{definition}
\defn{Interval-pair cover} is the following problem:

\textbf{Instance:} A universe $U = \{1,2,\dots,n\}$
and $m$ pairs of closed intervals $([a_i,b_i], [c_i,d_i])$ for $\newline i=1,2,\dots,m$. Here $a_i,b_i,c_i,d_i \in U$, $a_i \leq b_i$, and $c_i \leq d_i$.

\textbf{Question:} Is there a choice of one interval from each pair such that
every $i \in U$ is covered by some chosen interval?
\end{definition}

\begin{figure}
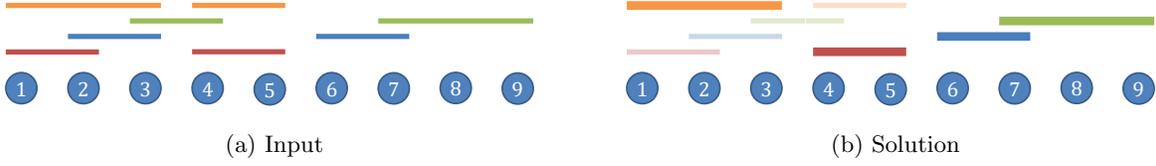

  \centering
  \subcaptionbox{Input}{\includegraphics[page=8,width=0.45\linewidth,trim=0 70 0 0,clip]{ppt_figs}}\hfil\hfil
  \subcaptionbox{Solution}{\includegraphics[page=9,width=0.45\linewidth,trim=0 70 0 0,clip]{ppt_figs}}
  \caption{Interval-Pair Cover: sample input and solution.
    The two intervals in the same pair are colored the same
    and share the same $y$ coordinate.}
  \label{interval pair cover}
\end{figure}

\begin{theorem} \label{IPC thm}
  Interval-Pair Cover is NP-complete, even when every interval pair
  $([a_j,b_j], [c_j,d_j])$ satisfies $a_j=b_j$ and $d_j - c_j \in \{0,1\}$.
\end{theorem}

\begin{proof}
We reduce from {\pospathSAT}; refer to Figure~\ref{IPC}.
We draw the shared-literal graph on the integer line from $1$ through $n$, placing the vertices at integer coordinates and using unit-length edges.  This is always possible because the shared-literal graph is a matching.  Then we create an interval pair for each variable $x_i$.  The pair's first interval contains only the coordinate of the vertex representing the clause where $x_i$ appears negatively; by the E1N property, there is exactly one.  The pair's second interval contains only the coordinate(s) of the vertex or vertices representing the clause(s) where $x_i$ appears positively; by the 2P property, there are at most two, and they are adjacent on the line because they share an edge in the shared-literal graph.  If $x_i$ does not appear positively, we set the second interval equal to the first interval.

\xxx{could get and use $\ge$1P property by eliminating any variables with zero positive occurrences as we do for zero negative occurrences}

\begin{figure}
  \centering
  \subcaptionbox{{\PospathSAT} instance}{\includegraphics[page=4,width=0.45\linewidth]{ppt_figs}}\hfil\hfil
  \subcaptionbox{Short drawing}{\includegraphics[page=5,width=0.45\linewidth]{ppt_figs}}

  \bigskip

  \subcaptionbox{Interval-pair cover instance}{\includegraphics[page=6,width=0.45\linewidth]{ppt_figs}}\hfil\hfil
  \subcaptionbox{Solution}{\includegraphics[page=7,width=0.45\linewidth]{ppt_figs}}
  \caption{Reduction from {\pospathSAT} to Interval-Pair Cover of Theorem~\ref{IPC thm}.}
  \label{IPC}
\end{figure}

The produced Interval-Pair Cover instance has a solution if and only if the input {\pospathSAT} instance is satisfiable.  Given a satisfying truth assignment, from the interval pair corresponding to variable $x_i$, we choose the first interval if $x_i$ is assigned false and the second interval if $x_i$ is assigned true.  Each chosen interval covers the coordinate(s) of the clause vertices satisfied by $x_i$, so if the truth assignment satisfies the formula, the chosen intervals cover all integers in the universe.  Given a complete interval cover, we assign true to $x_i$ if the second interval was chosen from its corresponding pair and false if the first interval was chosen.  By the same interval-variable correspondence, if the intervals cover all integers in the universe, the constructed truth assignment satisfies the formula.
\end{proof}

\subsubsection{Reduction from Interval-Pair Cover to $1 \times n$ $<$-Compatible Edge Matching}

\begin{theorem}
  $1 \times n$ $<$-compatible edge matching is NP-complete.
\end{theorem}

\begin{proof}
We reduce from Interval-Pair Cover.  For each integer $i$ in the Interval-Pair Cover universe $\{1, 2, \dots, n\}$, we create two copies of the element tile \TILE{i}{i}{i}{i}.  For each interval pair $([a_j, b_j], [c_j, d_j])$, we create an interval-pair tile \TILE{a_j-1}{c_j-1}{b_j{+}1}{d_j{+}1}.  The edge-matching board is $1 \times (2n+m)$, where $n$ is the size of the universe and $m$ is the number of interval pairs.

Given a solution to the produced edge-matching instance, we can construct a solution to Interval-Pair Cover by choosing each interval tile's horizontally-oriented interval (e.g., the interval $[a_j,b_j]$ for a tile oriented as \TILE{a_j-1}{c_j-1}{b_j{+}1}{d_j{+}1} or as \TILE{b_j{+}1}{d_j{+}1}{a_j-1}{c_j-1}). Suppose for contradiction that an element $i$ is uncovered by every chosen interval. Then in every placed tile whose left edge is at least $i+1$, its right edge is at least $i$, so the left edge of the next tile is at least $i+1$. In the sequence of left edges of tiles, the left edge of the tile after the first copy of \TILE{i}{i}{i}{i} is at least $i+1$, so every following left edge is at least $i+1$, leaving no place for the second copy of \TILE{i}{i}{i}{i}.

Given a solution to Interval-Pair Cover, we can construct a solution to the produced edge-matching instance. We will first describe a solution that uses extra copies of \TILE{i}{i}{i}{i}. For each chosen interval $[a_j,b_j]$, orient the tile as \TILE{b_j{+}1}{d_j{+}1}{a_j-1}{c_j-1}, and attach to its right \TILE{a_j}{a_j}{a_j}{a_j}, \ldots, \TILE{b_j}{b_j}{b_j}{b_j} to get a sequence of tiles with left edge $b_j+1$ and right edge $b_j$. Now place, for each $i \in \{1,2,\ldots,n\}$, the tile \TILE{i}{i}{i}{i}, followed by any of the above sequences of tiles with left edge $i+1$ and right edge $i$. That uses as many copies of \TILE{i}{i}{i}{i} as the number of intervals that cover $i$, plus 1, which is at least two. We can remove any \TILE{i}{i}{i}{i} and leave a valid solution, so arbitrarily removing copies until there are two copies of each \TILE{i}{i}{i}{i} left leaves a solution to the edge-matching instance.

\end{proof}

\section{$1 \times n$ Edge Matching ASP/\#P-completeness}
\label{sec:ASP}

In this section, we adapt the work of \cite{bosboom2017even} to show that $1\times n$ edge-matching puzzles are ASP- and \#P-complete.  Like \cite{bosboom2017even}, we reduce from Hamiltonian path in planar 3-regular directed graphs, which we newly prove ASP-~and~\#P-complete.

\subsection{Directed Hamiltonicity ASP/\#P-completeness}

Seta's thesis \cite{Seta02thecomplexities} proves ASP-completeness for Hamiltonicity in planar maximum-degree-3 \emph{undirected} graphs.
Here we prove the analogous result for \emph{directed} graphs:

\label{sec:Ham ASP}
\begin{theorem}
Finding Hamiltonian cycles in a planar 3-regular directed graph with maximum indegree 2 and maximum outdegree 2 is ASP-complete, and counting Hamiltonian cycles in those graphs is \#P-complete. %
\end{theorem}

\begin{proof}
These problems are clearly in FNP and \#P respectively.  To prove hardness,
we give a parsimonious reduction from (planar) positive 1-in-3SAT,
which is known to be ASP-complete and \#P-complete
\cite{DBLP:journals/siamcomp/HuntMRS98}.%
\footnote{Our proof does not actually use the planarity of the 1-in-3SAT
  instance.  To avoid the exclusive-or crossover gadget, we would need the
  variable-clause graph to remain planar with a line through all of the
  variables and all of the clauses, a variant not known hard
  \cite{ivan-thesis}.}
Our reduction is patterned after Plesn\'ik's NP-hardness reduction from 3SAT for Hamiltonian cycle in this class of graphs \cite{plesnik}.  Plesn\'ik's reduction does not conserve the number of solutions because the clause gadget admits multiple solutions when multiple literals in the clause are satisfied (Figure~\ref{fig:plesnik-clause}).  Reducing from 1-in-3SAT and simplifying Plesn\'ik's clause gadget allows us to conserve the number of solutions, and reducing from \emph{positive} 1-in-3SAT (no negated literals) allows us to simplify the clause gadget.  Plesn\'ik's exclusive-or gadget and exclusive-or crossover gadget do not give rise to additional solutions, so they can be used as-is. %

Figure~\ref{fig:hamiltonicity-full} shows a full Hamiltonicity instance produced by our reduction, with variable gadgets on the right (heading down) and clause gadgets on the left (heading up), and variables and clauses connected by exclusive-or lines (the green lines with hollow endpoints) which may cross.  (Compare \cite[Figure 1]{plesnik}, in which Plesn\'ik has abbreviated the clause gadgets.)

\begin{figure}[!t]
  \centering
  \includegraphics[height=.85\textheight]{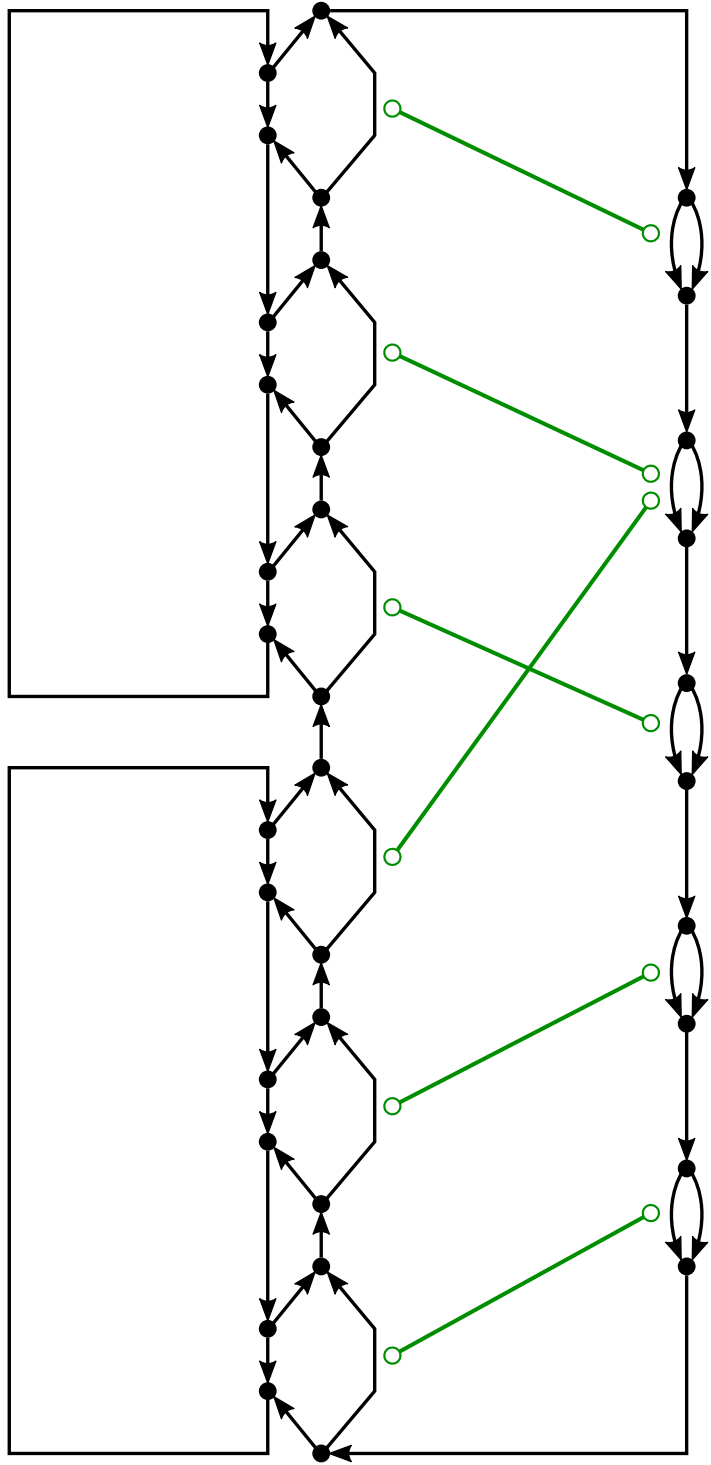}
  \caption{A full Hamiltonicity instance produced by our reduction, with variable gadgets on the right (heading down) and clause gadgets on the left (heading up).
    Variables and clauses are connected by exclusive-or lines (the green lines
    with hollow endpoints) as defined in Figure~\ref{fig:hamiltonicity-xor},
    with crossings expanded as in Figure~\ref{fig:hamiltonicity-xor-crossover}.}
  \label{fig:hamiltonicity-full}
\end{figure}

\paragraph{Exclusive-or line.}  An exclusive-or line between two edges abbreviates the pattern of additional vertices and edges shown in Figure~\ref{fig:hamiltonicity-xor}.  Traversing either of the two edges covers all of the additional vertices in exactly one way, excluding the other original edge from the cycle.  Traversing a path not corresponding to one of the original edges (e.g., from the bottom left to bottom right in Figure~\ref{fig:hamiltonicity-xor}) prevents the center four vertices from being part of any cycle (either they are uncovered, or they are the last four vertices in the path, so the path is not a cycle).  If neither of the two original edges is used, all of the additional vertices are uncovered.

\begin{figure}
\centering
\begin{minipage}{0.45\linewidth}
  \centering
  \begin{subfigure}{0.5\linewidth}
    \centering
    \includegraphics[scale=1.2]{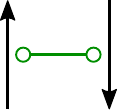}
  \end{subfigure}%
  \begin{subfigure}{0.5\linewidth}
    \centering
    \includegraphics[scale=1.2]{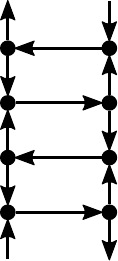}
  \end{subfigure}
  \caption{Our notation for an exclusive-or line between two edges and its expansion into additional vertices and edges.  (Redrawing of \cite[Figure 4]{plesnik}.)}
  \label{fig:hamiltonicity-xor}
\end{minipage}\hfill
\begin{minipage}{0.5\linewidth}
  \centering
  \begin{subfigure}{0.5\linewidth}
    \centering
    \includegraphics[scale=1.2]{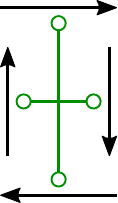}
    \label{fig:hamiltonicity-xor-crossover-notation}
  \end{subfigure}%
  \begin{subfigure}{0.5\linewidth}
    \centering
    \includegraphics[scale=1.2]{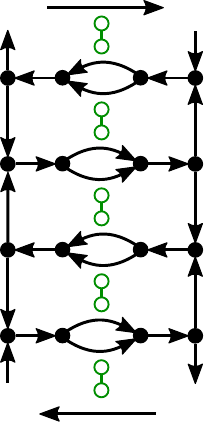}
    \label{fig:hamiltonicity-xor-crossover-expansion}
  \end{subfigure}
  \caption{Expansion of an exclusive-or line that crosses another exclusive-or line.  (Based on \cite[Figure 5]{plesnik}, simplified to show only two lines crossing.)}
  \label{fig:hamiltonicity-xor-crossover}
\end{minipage}
\end{figure}

\paragraph{Exclusive-or crossover.}  Exclusive-or lines connecting variable gadgets to clause gadgets may cross, necessitating the exclusive-or crossover shown in Figure~\ref{fig:hamiltonicity-xor-crossover}.  The crossover works by splitting each crossed-over edge between one pair of original edges into two edges and adding new exclusive-or lines that guarantee the parity of these paired edges is the same throughout the gadget.  For example, in Figure~\ref{fig:hamiltonicity-xor-crossover}, if the top edge is in the cycle, then the top edge of each pair is also in the cycle and the bottom edge is not in the cycle, regardless of which of the left or right edges are in the cycle.  As before, the expansion can be traversed in exactly one way for each pair of original edges traversed, and a traversal not corresponding to an original edge leaves some vertices uncovered.

\paragraph{Variable gadget.}  The variable gadget is a pair of vertices connected by a pair of parallel edges.\footnote{The graph is a simple graph, not a multigraph: If we remove any variables not used in any clauses, then for each variable, one of these edges will be replaced by an exclusive-or gadget, leaving no parallel edges.} The edge on the interior face of the variable-clause cycle is connected by exclusive-or lines to each clause in which the variable appears; including this edge in the Hamiltonian cycle represents setting the variable to true.  The other edge of the variable gadget is not connected to anything and represents setting the variable to false.  The variable gadgets are connected in sequence.

Plesn\'ik's variable gadget used two pairs of parallel edges, connected on the exterior by an exclusive-or line such that they have opposite settings, with the second pair connected to clauses where the variable appeared as a negative literal.  We reduce from planar \emph{positive} 1-in-3SAT, so all literals in our clauses are positive, making the second pair unnecessary.

\paragraph{Clause gadget.}  Our clause gadget and its three Hamiltonian paths are shown in Figure~\ref{fig:hamiltonicity-clause}.  The three rightmost edges in the clause gadget are connected by exclusive-or lines to the variable gadgets corresponding to the variables appearing in this clause.  If a variable is set to true, then the rightmost edge connected to that variable gadget cannot be in the cycle; otherwise, that rightmost edge \emph{must} be in the cycle.  If exactly one of the three variables is true, then the clause gadget can be covered in exactly one way (using one of the paths shown in Figure~\ref{fig:hamiltonicity-clause}).  If a variable is true, the path must go to the left of that hexagon, where it must enter the left loop.  If the path leaves the left loop before visiting all vertices in it, it cannot visit the top vertex of the hexagon where it entered the loop, so the left loop must be covered in its entirety.  But then the path cannot go left in any other hexagon, so the other variable must be false.  If all variables are false, the left loop is uncovered.  Thus this gadget simulates a 1-in-3SAT clause.

Our clause gadget differs from Plesn\'ik's by the deletion of the ``bridges'' between the hexagons and the left loop.  The bridges allow multiple literals to be simultaneously true, which is necessary for Plesn\'ik's reduction (from 3SAT).

\begin{figure}
    \centering
    \begin{subfigure}{\linewidth}
        \centering
        \includegraphics[height=7cm]{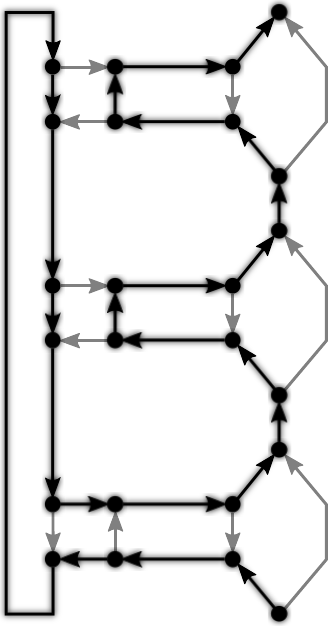}\hfil%
        \includegraphics[height=7cm]{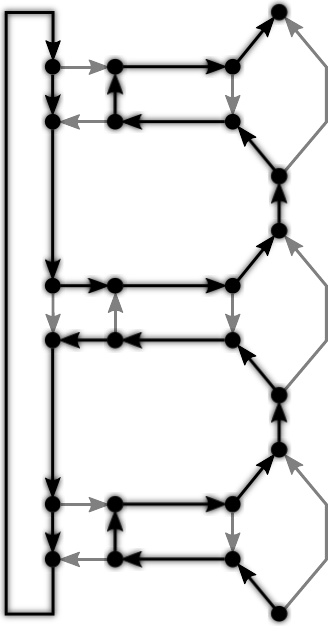}\hfil%
        \includegraphics[height=7cm]{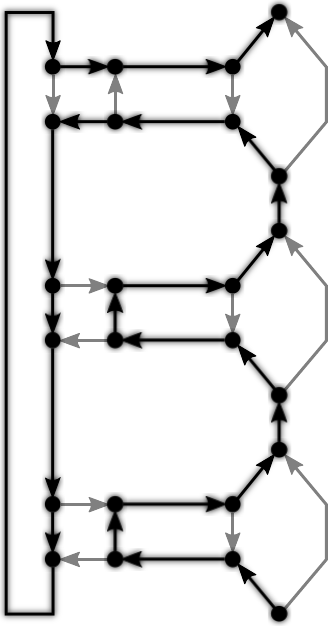}\hfil%
        \caption{The three Hamiltonian paths through Plesn\'ik's clause gadget \cite[Fig. 2]{plesnik} when all three literals are true.  (The right edge of each hexagon is covered via the exclusive-or line from the variable gadget.)}
        \label{fig:plesnik-clause}
    \end{subfigure}

    \begin{subfigure}{\linewidth}
        \centering
        \includegraphics[height=7cm]{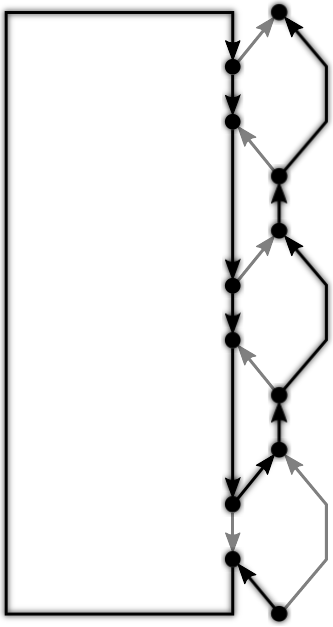}\hfil%
        \includegraphics[height=7cm]{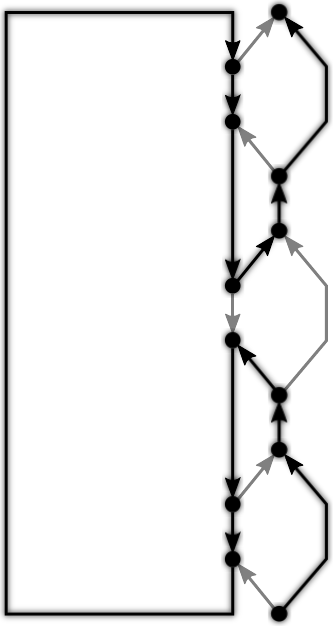}\hfil%
        \includegraphics[height=7cm]{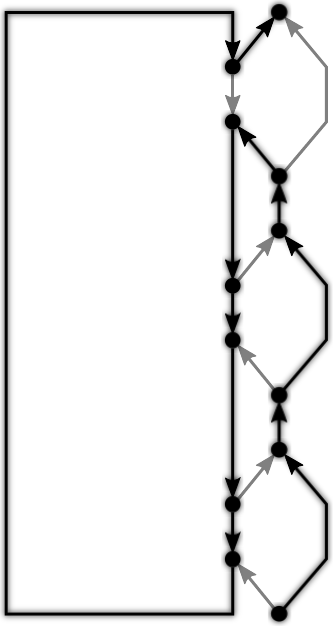}\hfil%
        \caption{The three Hamiltonian paths through our modified clause gadget.  The right edges of two of the three hexagons are always used, so this is a 1-in-3SAT clause.}
        \label{fig:hamiltonicity-clause}
    \end{subfigure}
    \caption{Comparison of Plesn\'ik's clause gadget and our modified clause gadget.}
    \label{fig:ham-clause}
\end{figure}

\paragraph{Conclusion.}  Figure~\ref{fig:hamiltonicity-full} shows a full instance produced by our reduction.  For each satisfying assignment of the variables, there is one corresponding Hamiltonian cycle using the corresponding configuration of the variable gadgets and the unique satisfying path through each clause gadget.  Conversely, a satisfying assignment can be uniquely read off from each Hamiltonian cycle based on the configuration of the variable gadgets.
\end{proof}

\begin{theorem} \label{cycle-to-path}
Finding Hamiltonian paths, with or without given start vertex $s$ and/or end vertex $t$, in planar 3-regular directed graphs with maximum indegree 2 and maximum outdegree 2 is ASP-complete,
and counting Hamiltonian paths in those graphs is \#P-complete.
The same result holds when the given vertex $s$ has outdegree 1 and the given vertex $t$ has indegree 1.
\end{theorem}

\begin{proof}
We prove this result via a parsimonious reduction from Hamiltonian cycle in planar 3-regular graphs with maximum indegree and outdegree 2. Given a 3-regular directed graph, we find an edge $uv$ that must be in every Hamiltonian cycle (an outgoing edge from a vertex with indegree 2, or an incoming edge to a vertex with outdegree 2).  We split $uv$, introducing two degree-1 vertices but otherwise leaving the graph 3-regular.

\begin{figure}
  \centering
  \begin{subfigure}{.3\linewidth}
    \centering
    \includegraphics[scale=2.0]{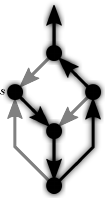}
    \caption{3-regular start vertex gadget}
    \label{fig:hamiltonicity-3regular-terminators-start}
  \end{subfigure}
  \hfil
  \begin{subfigure}{.3\linewidth}
    \centering
    \includegraphics[scale=2.0]{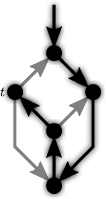}
    \caption{3-regular end vertex gadget}
    \label{fig:hamiltonicity-3regular-terminators-end}
  \end{subfigure}
  \caption{Gadgets that replace degree-1 start or end vertices to restore 3-regularity to the overall graph while maintaining a unique Hamiltonian path.  Vertices $s$ and $t$ are the new start and end vertices.}
  \label{fig:hamiltonicity-3regular-terminators}
\end{figure}

To restore 3-regularity we replace the degree-1 vertices with the graphs shown in Figure~\ref{fig:hamiltonicity-3regular-terminators}.  The unique longest (simple) path entering the graph in Figure~\ref{fig:hamiltonicity-3regular-terminators-end} ends at the vertex labeled $t$, because the first three vertices have outdegree 1 and the other successor of the fourth vertex is already in the path.  By a similar argument working backwards from the outgoing edge of the graph in Figure~\ref{fig:hamiltonicity-3regular-terminators-start}, the unique longest path leaving the graph starts at the vertex labeled $s$.  Thus, whether or not $s$ and $t$ are specified as the start and end vertices in the Hamiltonian path instance, all Hamiltonian paths in the transformed graph start at $s$ and end at $t$.  Vertex $s$ has outdegree 1 and $t$ has indegree 1, as claimed in the theorem statement.  Because $uv$ occurs in every Hamiltonian cycle of the input graph, there is a bijection between Hamiltonian cycles in the input instance and Hamiltonian paths in the output instance, and this bijection can be computed in polynomial time by replacing $uv$ with the unique paths in the start/end gadgets or vice versa.
\end{proof}

\subsection{Reduction from Hamiltonicity to $1 \times n$ Edge Matching}

The symmetry of $1\times n$ edge-matching puzzles is problematic for ASP-hardness. Because rotating any solution by $180^\circ$ will give another solution, the answer to the ASP problem is always `yes'. To avoid this trivial additional solution, we consider the version of $1\times n$ edge-matching puzzles where the left boundary edge's color is specified. This breaks the rotational symmetry, and we will show that this problem is ASP-complete through a parsimonious reduction. Without this restriction, our reduction is 2-monious, so we show \#P-hardness even for $1\times n$ edge-matching puzzles without any such restriction.

The reduction in \cite{bosboom2017even} that establishes NP-hardness of $1\times n$ edge-matching puzzles is not parsimonious because of garbage collection: the tiles corresponding to edges which are not part of the Hamiltonian path are placed at the end of the row of tiles in an arbitrary order. Our reduction will instead place these unused tiles near the corresponding vertex tiles, so that there is only one tile sequence corresponding to each Hamiltonian path.

\begin{theorem}\label{thm:asp}
    $1\times n$ signed and unsigned edge-matching puzzles with the left boundary edge color specified are ASP-complete and \#P-complete.
\end{theorem}

\begin{proof}
    Clearly this problem is in FNP and its counting problem is in \#P. To show hardness, we present a parsimonious reduction from Hamiltonian path in 3-regular directed graphs, adapted from the reduction in \cite{bosboom2017even}.
    
    Given a 3-regular directed graph $G$ with specified vertices $s$ and $t$, we construct a $1\times n$ signed edge-matching puzzle as follows.
    (For the unsigned case, we will simply remove all signs.)
    For each edge $e$ in $G$, we have a color $e$, and for each vertex $v$ we have three colors $v_I$, $v_O$, and $v_X$.
    For each vertex $v$, we build three tiles; refer to Figure~\ref{fig:asp}.
    In one case, $v$ has one edge $e_1$ coming in and two edges $e_2$ and $e_3$ going out. Then we construct the tiles
    $$\TILE{+e_1}{-v_X}{-v_I}{-v_X}, \quad  \TILE{+v_O}{-v_O}{-e_2}{+v_I},
    \quad \text{and} \quad \TILE{+v_O}{-v_O}{-e_3}{+v_I}.$$
    In the other case, $v$ has two edges $e_1$ and $e_2$ coming in and one edge $e_3$ going out. Then we construct the tiles
    $$\TILE{+e_1}{+v_I}{-v_I}{-v_O}, \quad \TILE{+e_2}{+v_I}{-v_I}{-v_O}
    \quad \text{and} \quad \TILE{+v_O}{-v_X}{-e_3}{-v_X}.$$
    Each of these tiles corresponds to one of the half-edges incident to~$v$.
    (Overall, each edge is represented by two half-edge tiles.)
    We use that $s$ has outdegree 1 and $t$ has indegree 1, as provided by Theorem~\ref{cycle-to-path}. We remove the tiles corresponding to the half-edges entering $s$ and the tiles corresponding to half-edges leaving $t$, so $s$ and $t$ each have only one corresponding tile. Finally, we specify that the left boundary edge has color~$-s_O$.
    
    \begin{figure}
        \centering
        \begin{subfigure}{.3\linewidth}
            \centering
            \includegraphics[width=\linewidth]{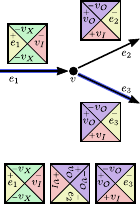}
            \caption{A vertex with indegree 1 and outdegree 2.}
        \end{subfigure}
        \hfil
        \begin{subfigure}{.3\linewidth}
            \centering
            \includegraphics[width=\linewidth]{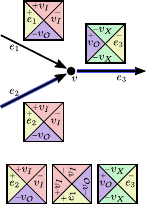}
            \caption{A vertex with indegree 2 and outdegree 1.}
        \end{subfigure}
        \caption{The tiles in the reduction showing ASP- and \#P-hardness of $1\times n$ edge-matching puzzles. At the bottom we show one possible edge-matching solution corresponding to one (blue) path through~$v$.}
        \label{fig:asp}
    \end{figure}
    
    We claim that the number of solutions to this edge-matching puzzle is the
    same as the number of Hamiltonian paths in $G$ from $s$ to $t$.

    First suppose that we have such a Hamiltonian path
    $s=v_1,v_2,\dots,v_{|V|}=t$.
    We can construct a solution to the edge-matching puzzle by placing the
    three tiles for each vertex $v_i$ consecutively,
    in the order $i = 1, 2, \dots, |V|$ the vertices appear in the path.
    As in the bottom of Figure~\ref{fig:asp},
    we place the three tiles for each vertex $v_i$ so that the tiles
    corresponding to the edges $e_i = (v_{i-1},v_i)$ and
    $e_{i+1} = (v_i,v_{i+1})$ that the path uses to enter and exit $v$
    are first and last, respectively, so the sequence of colors is
    $e_i,v_{i,I},v_{i,O},e_{i+1}$.
    The exposed colors are $+e_i$ on the left and $-e_{i+1}$ on the right,
    so the these placed triples of tiles match up at their ends
    (because the sequence of vertices is a path).
    There is only one tile for each of $s$ and $t$, which we place at the
    beginning and end.
    The left boundary color is then $+s_O$, as required,
    and the rightmost boundary color is $+t_I$.
    
    Next we show that every solution to the edge-matching puzzle has this form, and thus corresponds to a Hamiltonian path. Suppose we have a solution to the edge-matching puzzle. Because the left boundary color is $-s_O$, the tile corresponding to $s$ must be placed on the left oriented with $+s_O$ on the left and the outgoing edge color on the right.  The only tile corresponding to $t$ is $\TILE{+e}{-t_X}{-t_I}{-t_X}$, where $e$ is the incoming edge. Because colors $t_X$ and $t_I$ do not appear on any other tiles, this tile must be placed rightmost with color $+e$ on the left.
    
    Consider a vertex $v$ other than $s$ and $t$. None of the tiles corresponding to $v$ can be at either end of the solution, because those spaces are claimed by $s$ and $t$. Suppose $v$ has indegree 1 and outdegree 2; the other case is similar. Because $\TILE{+e_1}{-v_X}{-v_I}{-v_X}$ is the only tile with the color $v_X$, it must be adjacent to other tiles on the other two sides. The tile adjacent on the side with color $-v_I$ must be one of the two other tiles corresponding to~$v$. Whichever tile it is, its orientation is fixed by matching color $v_I$, so the opposite side must have color $-v_O$, so the following tile must be the third tile corresponding to $v$, with the color of another edge incident to $v$ on the side touching the next tile. In summary, the three tiles corresponding to $v$ must be consecutive, and the two colors they expose to other tiles are two edges incident to $v$ with different orientations relative to $v$, with the local configuration of the three tiles determined by those exposed colors.
    
    Suppose the sequence of tiles corresponding to vertex $u$ are adjacent to the sequence corresponding to vertex $v$. Then the side where these sequences touch must have color $e$, where $e$ is either $(u,v)$ outgoing from $u$ and incoming to $v$ or $(v,u)$ outgoing from $v$ and incoming to~$u$. The other left and right edges of these tiles must also have edge colors corresponding to edges incident to $u$ and $v$. By induction, if the solution has several consecutive sequences of tiles corresponding to vertices, the sequence of vertices must form a path in $G$ in either direction. The entire solution must therefore be a concatenation of sequences corresponding to vertices starting with $s$ and ending with $t$, such that adjacent vertices share an edge from left to right, and using each tile exactly once. Hence it must correspond to a Hamiltonian path.
    
    For each Hamiltonian path, there is exactly one corresponding solution to the edge-matching puzzle, because there is only one way to connect the tiles corresponding to a vertex for each pair of edges used at that vertex. So there are the same number of Hamiltonian paths in $G$ from $s$ to $t$ and solutions to the edge-matching puzzle. Because this reduction is parsimonious, it shows that $1\times n$ signed edge-matching puzzles with the color of the left boundary edge specified is ASP- and \#P-complete.  The same reduction with all the signs removed proves the same result for unsigned edge-matching puzzles.
\end{proof}

\begin{corollary}\label{cor:asp}
    $1\times n$ signed and unsigned edge-matching puzzles are \#P-complete and their
    2-ASP problem is NP-hard.
\end{corollary}

\begin{proof}
Without a specified left boundary color, we cannot guarantee that the tile corresponding to the start vertex $s$ is on the left and the tile corresponding to the end vertex $t$ is on the right; instead we only have that they are at the ends.  Thus each solution to the edge-matching puzzle can be rotated $180^\circ$ to form another solution, so the reduction is 2-monious.
\end{proof}

\section{Triangular Edge Matching}
\label{sec:triangle}

In this section, we study $1 \times n$ edge-matching puzzles with triangular tiles, specifically, equilateral and right isosceles triangles.
There is one natural interpretation of ``$1 \times n$'' for equilateral triangles, as shown in Figure~\ref{fig:equilateral}.
However, for right isosceles triangles, there are two natural interpretations.
If the triangles have legs of length 1, then to pack a $1 \times n$ box they must have alternating hypotenuse/leg contact, which we will simply refer to as hypotenuse contact, as shown in Figure~\ref{fig:hypotenuse}. On the other hand, if the triangles have a height of 1, then they must be packed using only leg-to-leg contacts, as shown in Figure~\ref{fig:leg}.

\begin{figure}[h]
  \centering
  \subcaptionbox
    {\label{fig:equilateral} Equilateral triangles}
    {\includegraphics[scale=0.6]{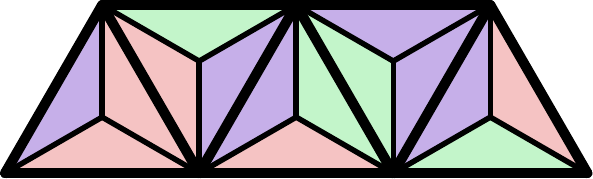}}

  \bigskip

  \subcaptionbox
    {\label{fig:hypotenuse} Right triangles, hypotenuse contact}
    {\includegraphics[scale=0.6]{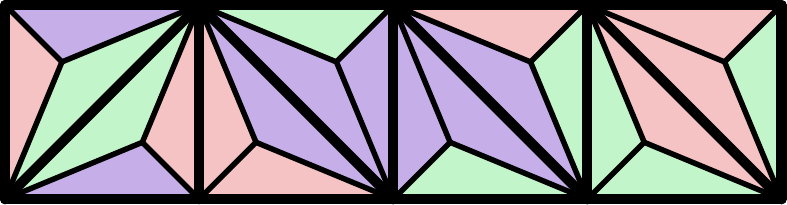}}
  \hfil\hfil
  \subcaptionbox
    {\label{fig:leg} Right triangles, leg contact}
    {\includegraphics[scale=0.6]{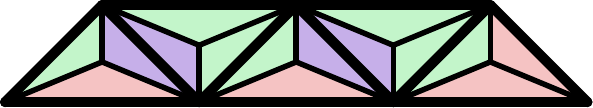}}

  \caption{Three types of triangular tiles.}
  \label{fig:right triangle placements}
\end{figure}

Hypotenuse-contact right triangles can directly and parsimoniously simulate square tiles: for each square, create two triangles whose hypotenuses have a matching, unique color.
(This idea is mentioned in another context in the conclusion of \cite{bosboom2017even}.)
Thus NP-completeness, ASP-completeness (with left boundary specified), and \#P-completeness of these puzzles follows directly from results on square tiles.
We devote the rest of this section to
equilateral triangles (Section~\ref{equilateral-triangle-hardness})
and right triangles with leg contact (Section~\ref{right triangles}).

\subsection{Equilateral-Triangle Edge Matching} \label{equilateral-triangle-hardness}

In this section, we prove NP/\#P/ASP-completeness of $1 \times n$ equilateral
triangular edge-matching puzzles. We start with an NP-completeness proof, then
augment it and analyze it further to prove \#P/ASP-completeness.

\begin{theorem} \label{thm:equilateral NP}
  $1\times n$ signed and unsigned equilateral-triangle edge-matching puzzles
  are NP-complete.  The same results hold if we allow tile reflection.
\end{theorem}

\begin{proof}
Clearly these problems are in NP.
To show NP-hardness, we reduce from Hamiltonian path in 3-regular
\emph{undirected} graphs \cite{Garey-Johnson-1979}
(in contrast to Section~\ref{sec:ASP} which considered directed graphs).
We describe signed tiles resulting from our reduction to signed edge matching;
for the unsigned puzzle, we will just drop the signs.
Similar to the proof of Theorem~\ref{thm:asp}, we will create exactly two tiles
per edge; refer to Figure~\ref{signed triangle hardness}.
To assign complementary signs to the edge colors,
arbitrarily orient each edge $e$ (but paths need not follow this orientation).
For every vertex $v$ with incident edges $e_1, e_2, e_3$,
construct the triangular tiles
$$\TRITILEEQ{+v}{-v}{\pm e_1}, \quad \TRITILEEQ{+v}{-v}{\pm e_2},
\quad \text{and} \quad \TRITILEEQ{+v}{-v}{\pm e_3},$$
where the sign of each $e_i$ color is positive if $e_i$ was arbitrarily
oriented to be incoming to $v$ and negative otherwise.
We claim that these tiles have a signed or unsigned edge-matching solution
if and only if the graph has a Hamiltonian path.

\begin{figure}
\centering
\includegraphics[scale=2]{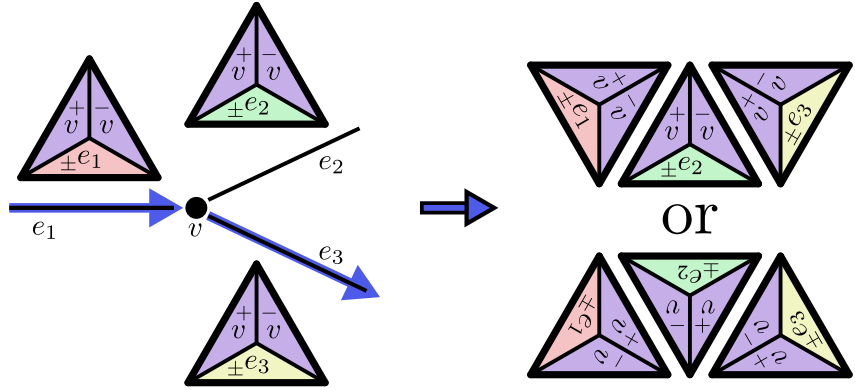}
\caption{NP-hardness of $1 \times n$ equilateral-triangle edge matching, showing one possible (blue) path through $v$ and the corresponding edge-matching solutions (depending on parity up to this point).}
\label{signed triangle hardness}
\end{figure}

First suppose that there is a Hamiltonian path $v_1, v_2, \ldots, v_n$.
We can construct an edge-matching solution by, for each vertex $v_i$,
arranging the three corresponding tiles so that $\{v_{i-1},v_i\}$ is on the
left boundary edge and $\{v_i,v_{i+1}\}$ is on the right boundary edge,
as in Figure~\ref{signed triangle hardness} (right, top or bottom according to
parity of $i$ as required by the tiling).
The figure illustrates that the vertex colors match with opposite signs,
and that tiles do not need to be reflected.
By the arbitrary orientation of the edges, every edge color will match with
its negated color.

Now suppose that there is an edge-matching solution, even without the color
signs.  Without the color signs, the tiles are reflectionally symmetric,
so the following argument works also when we allow tile reflection.
Each vertex color $v$ appears in exactly three tiles, so the three vertex tiles
can match only with each other, or some of them can appear as the extreme left
or extreme right tile.  If any of the three tiles for $v$ are extreme, then
none of the tiles can be placed in the middle of an edge-matching solution
(lacking the three tiles required to form an $180^\circ$ angle), so in this
case, all three tiles for $v$ appear at the left and right extremes of the
solution, effectively ``wrapping around'' the $1 \times n$ board.
For every other vertex, the three corresponding tiles must appear together.
Listing all of the vertices in the order in which their color appears in the
solution yields a Hamiltonian path of the original graph.
(If one vertex's tiles wrap around, then this process yields a Hamiltonian
cycle, which is stronger.)
\end{proof}

The proof above suggests an alternate approach to proving Theorem~\ref{thm:asp}
about squares: unify the $v_I, v_O, v_x$ colors into a single color,
and reduce from \emph{undirected} Hamiltonian path.  However, for unsigned
colors, this change would make the reduction not parsimonious, because it
enables the middle tile to rotate by $180^\circ$ in the two arrangements on the
bottom of Figure~\ref{fig:asp}.
But equilateral triangles lack this ambiguity, and we are able to obtain
parsimony by a more careful handling of the start and end.

First we need a slightly different form of undirected Hamiltonicity:

\begin{lemma} \label{lem:undirected Ham}
  Finding Hamiltonian paths, with or without specified start vertex $s$ and/or
  end vertex $t$, in maximum-degree-3 planar undirected graphs is ASP-complete,
  and counting Hamiltonian paths in those graphs is \#P-complete.
  The same result holds when the given vertices $s,t$ have degree~1.
\end{lemma}

\begin{proof}
We present a parsimonious reduction from Hamiltonian cycle in maximum-degree-3
planar undirected graphs (the same graphs) having at least one vertex of
degree $2$, proved ASP-complete by Seta \cite{Seta02thecomplexities}.
Our reduction is similar to the first step in the proof of
Theorem~\ref{cycle-to-path}.

Let $G$ be a maximum-degree-3 undirected graph with a degree-$2$ vertex~$v$.
Let $\{u,v\}$ be one of $v$'s incident edges, which must be in every
Hamiltonian cycle.
Construct $G'$ by adding two new vertices $s$ and~$t$,
and replacing the edge $\{u,v\}$ with edges $\{s,u\}$ and $\{t,v\}$.
Because $s$ and $t$ have degree~$1$, they are in every Hamiltonian path of~$G'$.
Because edge $\{u,v\}$ is contained in every Hamiltonian cycle in $G'$,
there is a direct bijection between Hamiltonian cycles in $G$ and
Hamiltonian ($s$-$t$) paths in~$G'$.
\end{proof}

\begin{theorem}
  $1\times n$ signed and unsigned equilateral-triangle edge-matching puzzles
  with the left boundary edge color specified are ASP-complete and \#P-complete.
\end{theorem}

\begin{proof}
Clearly this problem is in FNP and its counting problem is in \#P.
To show hardness, we present a parsimonious reduction from
Hamiltonian $s$-$t$ paths in maximum-degree-3 undirected graphs
where $s$ and $t$ have degree~$1$, from Lemma~\ref{lem:undirected Ham}.
Our reduction is a modification of the NP-hardness reduction in
Theorem~\ref{thm:equilateral NP} that differs only for the new case of
vertices with degree $< 3$.
For each degree-$2$ vertex $v$, we attach a half-edge $\{v\}$
(with no other endpoint), and then apply the degree-$3$ construction from
Figure~\ref{signed triangle hardness}.
We can assume that the only degree-$1$ vertices are $s$ and $t$, because any
other degree-$1$ vertices could not possibly be reached by an $s$-$t$ path
(and thus in this case we could parsimoniously reduce by constructing any
unsolvable edge-matching instance).
For the degree-$1$ vertices $s$ and $t$, we create corresponding tiles
$$ \TRITILER{s}{\pm e_1}{U} \quad\text{and}\quad \TRITILER{\pm e_2}{U}{U}, $$
where $e_1$ and $e_2$ represent the unique edges incident to $s$ and $t$
respectively, with signs chosen for these edge colors based on our arbitrary
orientation of the original graph, in the same fashion as for all other tiles.
(As before, for the unsigned problem, we just drop the signs.)
Each occurrence of $U$ represents a unique color not occurring
in any other tile.
Finally, we specify the left boundary color to be~$s$,
which is another unique color.

Because the tile corresponding to vertex $s$ is the only one with color $s$,
it must be placed as the leftmost tile.
Because the tile corresponding to $t$ has two sides with unique colors,
it must be placed as the rightmost tile.
As argued in Theorem~\ref{thm:equilateral NP}, every triplet of tiles
corresponding to a degree-3 (or degree-2) vertex must occur consecutively,
because the $s$ and $t$ tiles prevent ``wrapping around''.
Therefore every edge-matching solution induces an ordering of the
vertex tile triplets between the leftmost $s$ tile and the rightmost $t$ tile.
To guarantee a bijection between edge-matching solutions and
Hamiltonian $s$-$t$ paths solutions, it only remains to show that,
given an ordering of the tile triplets, there is a unique arrangement
of the three tiles within each triplet.

Suppose the tile triplet for vertex $v$ occurs between the triplets for
vertices $u$ and $w$.  The only edge colors that $v$'s triplet have in common
with $u$'s and $w$'s triplets are the colors representing edges $\{u,v\}$ and
$\{v,w\}$ in the original graph, so the two tiles in $v$'s triplet containing
the $\{u,v\}$ and $\{v,w\}$ colors must be on the left and right respectively,
with those edges exposed.
The remaining tile in the triplet has no choice but to be oriented between
them with its $v$-colored edges facing the two other tiles in the triplet,
and its third edge facing the $1 \times n$ boundary.
Thus the arrangement of tiles within each triplet is uniquely defined by the
ordering of tile triplets along the box,
completing the proof that our reduction is parsimonious.
\end{proof}

\begin{corollary}\label{cor:equilateral asp}
    $1\times n$ signed and unsigned equilateral-triangle edge-matching puzzles
    are \#P-complete and their 2-ASP problem is NP-hard.
\end{corollary}

\begin{proof}
    As in Corollary~\ref{cor:asp}.
\end{proof}

\subsection{Leg-Contact Right-Isosceles-Triangle Edge Matching}
\label{right triangles}
In this section, we show that edge matching with right isosceles triangles
that tile a $1 \times n$ box by leg contact (as in Figure~\ref{fig:leg})
is closely related to finding an Eulerian path in a graph,
or more precisely, two variants called antidirected and forbidden-transition
Eulerian paths analyzed in Sections~\ref{sec:alternating euler}
and~\ref{sec:forbidden euler} respectively.
We use this connection to show that these puzzles can be solved in
polynomial time (Section~\ref{right triangle P}), and then to show that
counting solutions to these puzzles is \#P-complete
(Section~\ref{right triangle count}).

\subsubsection{Antidirected Eulerian Path Characterization}
\label{sec:alternating euler}

Consider a directed graph $G$.
Recall that a \defn{(directed) Eulerian path} is a directed path in $G$
(respecting the edge directions in~$G$)
that visits every edge in $G$ exactly once.  
It is well-known that a connected graph has such a path if and only if
it has zero or two vertices of odd degree
\cite[Corollary~4.1]{Bondy-Murty-1976},
and in this case, the path can be constructed in linear time
\cite{Fleischner-1991}.

Here we analyze the variant where the edge directions of $G$ must alternate.
Precisely, an \defn{antidirected path}
\cite{Grunbaum-1971,Berman-1978,antistrong} is a
sequence of edges where every pair of consecutive edges share an endpoint
(an undirected path) and furthermore those edges 
either both point toward or both point away from that shared endpoint.
In other words, an antidirected path alternates between following an
edge of $G$ in the ``forwards'' direction and following an edge of $G$
in the ``backwards'' direction, with an arbitrary starting parity.
An \defn{antidirected Eulerian path}
\cite{Berman-1978,Fleischner-1990,Zitnik-1996}
of $G$ is an antidirected path of $G$
that visits every edge (either forwards or backwards) exactly once.
Examples of past results on this topic include that
a directed graph without degree-2 vertices has an odd number of Eulerian paths
if and only if it is 4-regular and has an antidirected Eulerian path
\cite{Berman-1978},
while not every connected 4-regular undirected graph with an odd cycle has an
orientation admitting an antidirected Eulerian path \cite{Zitnik-1996}.

In the \defn{[antidirected] Eulerian path problem}, we are given a
directed graph $G$, and want to know whether $G$ has an [antidirected]
Eulerian path, and if so, to find one.
We relate these two problems:

\begin{theorem} \label{thm:alternating euler}
  The antidirected Eulerian path problem can be reduced
  in linear time to the Eulerian path problem.
\end{theorem}

\begin{proof}
Let $G$ be a directed graph input for the antidirected Eulerian path problem.
Construct an undirected bipartite graph $G'$
(called the ``split'' of $G$ by West \cite[Definition 1.4.20]{intro-to-graphs})
as follows; refer to Figure~\ref{fig:alternating euler}.
For each vertex $v \in G$, construct two vertices $v^+$ and $v^-$ in $G'$.
For every directed edge $e = (u, v) \in G$,
add the undirected edge $e' = \{u^+, v^-\}$ to $G'$.
Because every edge in $G'$ connects a plus vertex to a minus vertex,
$G'$ is bipartite.

\begin{figure}
\centering
\includegraphics[scale=0.75]{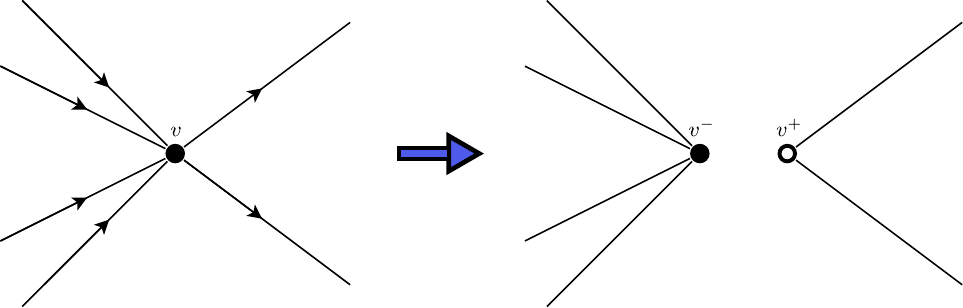}
\caption{Reduction from antidirected Eulerian path to Eulerian path.}
\label{fig:alternating euler}
\end{figure}

We claim that paths in $G'$ correspond to antidirected paths in~$G$.
For any path
$p' = (v_1^\pm, v_2^\mp, \allowbreak v_3^\pm, v_4^\mp, \allowbreak \dots)$
in $G'$ (where signs alternate by bipartiteness),
consider mapping each edge of the form $\{v_i^+, v_{i+1}^-\}$ in $p'$
to the corresponding edge $(v_i, v_{i+1})$ of~$G$, and
mapping each edge of the form $\{v_i^-, v_{i+1}^+\}$ in $p'$
to the (backwards traversal of) the corresponding edge $(v_{i+1},v_i)$ of $G$.
Then we obtain an antidirected path in~$G$.
Because the mapping between edges of $G$ and $G'$ is a bijection,
so is this transformation.
By the same bijectivity, if $p'$ is Eulerian, then so is~$p$.
Therefore Eulerian paths in $G'$ correspond to
antidirected Eulerian paths in~$G$.
\end{proof}

A similar result was obtained independently in
\cite{alternating-euler-independent}.

For our application to edge matching, we will need to solve a slightly
stronger version of the problem:

\begin{corollary} \label{cor:alternating euler}
  The antidirected Eulerian path problem can be solved in linear time.
  The same result holds if the path is further restricted to start and/or
  end with specified direction (forwards or backwards).
\end{corollary}

\begin{proof}
  The first sentence follows from the reduction of
  Theorem~\ref{thm:alternating euler} combined with linear-time algorithms
  for finding Eulerian paths \cite{Fleischner-1991}.

  Now suppose we are given the starting and ending directions
  $s, t \in \{\text{forwards}, \text{backwards}\}$ for an
  antidirected Eulerian path.
  Applying the previous algorithm, we can detect whether $G$ has any
  antidirected Eulerian path, i.e., whether $G'$ from the proof of
  Theorem~\ref{thm:alternating euler} has any Eulerian path.
  If the answer is ``no'', then we are done.
  Otherwise, by the characterization of Eulerian paths
  \cite[Corollary~4.1]{Bondy-Murty-1976},
  either (1)~every vertex of $G'$ has even degree,
  or (2)~exactly two vertices of $G'$ have odd degree.

  In the first case, every Eulerian path $p'$ of $G'$ is also a cycle, so when
  we translate to an antidirected Eulerian path/cycle $p$ of~$G$,
  the starting orientation is the same as the ending orientation
  if and only if $G$ has an odd number $e$ of edges.
  Thus we can answer the restricted antidirected Eulerian path problem
  by checking whether $(s = t) \leftrightarrow (e \text{ odd})$.
  If $s=t$ and $e$ is odd, then we find an antidirected Eulerian cycle and
  choose the starting parity for a path to match $s=t$.
  If $s \neq t$ and $e$ is even, then we find any antidirected Eulerian cycle
  and any starting point, and reverse the path if $s$ and $t$ mismatch.
  Otherwise, no satisfying antidirected Eulerian path exists.

  In the second case, every Eulerian path $p'$ of $G'$ has its endpoints at the
  two odd-degree vertices $o_1,o_2$ of~$G'$, so every antidirected Eulerian
  path $p$ in $G$ has its extreme edge orientations determined by whether
  $o_1$ and $o_2$ are plus or minus vertices (and which is chosen to be the
  start versus end of the path).
  If $o_1$ and $o_2$ are both plus vertices, then $s=\text{forwards}$ and
  $t=\text{backwards}$ is the only possibility.
  If $o_1$ and $o_2$ are both minus vertices, then $s=\text{backwards}$ and
  $t=\text{forwards}$ is the only possibility.
  If $o_1$ and $o_2$ are plus and minus vertices,
  then $s = t$ is the only constraint:
  if $s = t = \text{forwards}$, then we start at the plus vertex; and
  if $s = t = \text{backwards}$, then we start at the minus vertex.
  Otherwise, no satisfying antidirected Eulerian path exists.

  Therefore we can solve the restricted form of the
  antidirected Eulerian path problem.
\end{proof}

\subsubsection{Forbidden-Transition Eulerian Path Characterization}
\label{sec:forbidden euler}

In the \defn{forbidden-transition Eulerian path problem} \cite{Kotzig1968},
we are given an undirected graph $G = (V,E)$ and, for every vertex $v \in V$,
a partition of the edges $E_v$ incident to $v$ into groups
$P_{v,1}, P_{v,2}, \dots, P_{v,k_v}$.
The goal is to find an Eulerian path $v_0, v_1, \dots, v_{|E|}$ of $G$
such that, for every vertex visit $v_i$ where $0 < i < |E|$,
the incident edges $(v_{i-1},v_i)$ and $(v_i,v_{i+1})$ belong to
different groups among $P_{v_i,1}, P_{v_i,2}, \dots, P_{v_i,k_{v_i}}$.
In other words, we forbid the subpath $(v_{i-1},v_i,v_{i+1})$
when $(v_{i-1},v_i)$ and $(v_i,v_{i+1})$ belong to a common group $P_{v_i,j}$.%
\footnote{It is tempting to think that antidirected Eulerian path in a directed
  graph is a special case of forbidden-transition Eulerian path in an
  undirected graph, using two groups at each vertex to represent the outgoing
  vs.\ incoming edges.  However, the antidirected constraint \emph{requires
  repeating} the incoming/outgoing nature at each vertex, while the
  forbidden-transition constraint \emph{prevents repeating} the
  incoming/outgoing nature at each vertex.}
In a \defn{forbidden-transition Eulerian cycle}, we similarly restrict
the subpath $(v_{|E|-1}, v_{|E|} = v_0, v_1)$.

Kotzig \cite{Kotzig1968} showed (in a slightly more general scenario)
that the natural necessary conditions for this problem are in fact sufficient.
We repeat Kotzig's mathematical argument here in order to verify that it
also yields an efficient algorithm.

\begin{theorem}[\cite{Kotzig1968}] \label{thm:forbidden Euler}
  An undirected graph $G$ and partition system $P$ has a forbidden-transition
  Eulerian path if and only if $G$ has an Eulerian path and
  every group $P_{v,i}$ has $|P_{v,i}| \leq \lceil \degree(v)/2 \rceil$.
  If furthermore $G$ has an Eulerian cycle, then $(G,P)$ has a
  forbidden-transition Eulerian cycle.
  When such a path/cycle exists, it can be found in linear time.
\end{theorem}

\begin{proof}
  By the characterization of Eulerian paths
  \cite[Corollary~4.1]{Bondy-Murty-1976},
  $G$ must have exactly zero or two vertices of odd degree.
  We can reduce to the case of zero odd-degree vertices as follows.
  If $G$ has two odd-degree vertices, then add an edge between them,
  which increases their degrees to even but does not
  change $\lceil \degree(v_i)/2 \rceil$.
  Now apply the zero-odd-degree-vertices case of the present theorem
  (proved below) to obtain an Eulerian cycle with the desired property.
  Removing the added edge results in an Eulerian path with the desired property.
  Therefore we can assume every vertex has even degree,
  so we can ignore the ceilings.

  Next we prove that the conditions are necessary.
  Clearly $G$ having an Eulerian path is necessary for it to have a
  forbidden-transition Eulerian path.
  If any $|P_{v,i}| > \degree(v)/2$, then we claim that $(G,P)$ cannot have a
  forbidden-transition Eulerian path.  Any Eulerian path in $G$ is a cycle, and
  thus its traversal order pairs up the edges $E_v$ incident to $v$
  into $\degree(v)/2$ pairs.  By the Pigeonhole Principle,
  some pair has both its edges in $P_{v,i}$, which is a forbidden transition.

  Now suppose $G$ has an Eulerian path and every group $P_{v,i}$ satisfies
  $|P_{v,i}| \leq \degree(v)/2$.
  For each vertex $v$, order its incident edges
  $E_v = \{e_1, e_2, \dots, e_{\degree(v)}\}$
  so that all edges from group $P_{v,i}$ appear consecutively in the ordering,
  for all $1 \leq i \leq k_v$.
  Now pair each edge $e_j$ with $e_{j + \degree(v)/2}$,
  for $1 \leq j \leq \degree(v)/2$.
  Because each $|P_{v,i}| \leq \degree(v)/2$,
  this pairing has no forbidden pairs.
  The perfect pairing at each vertex partitions the graph's edges into
  edge-disjoint cycles.

  To merge these cycles into one Eulerian cycle,
  take any two cycles $C,C'$ that share a vertex~$v$
  (which exist because $G$ has an Eulerian path so its edges are connected).
  Suppose one cycle pairs edges $(e_1,e_2)$ at~$v$,
  while the other cycle pairs edges $(e'_1,e'_2)$ at~$v$.
  Suppose $e_1,e_2,e'_1,e'_2$ are in groups $i_1,i_2,i'_1,i'_2$.
  If we change the local pairing to $(e_1,e'_2)$ and $(e'_1,e_2)$,
  then we merge the cycles, and avoid forbidden pairs provided $i_1 \neq i'_2$
  and $i'_1 \neq i_2$.
  If we change the local pairing to $(e_1,e'_1)$ and $(e'_2,e_2)$
  (and reverse one of the cycles), then we again merge the cycles, this time
  avoiding forbidden pairs provided $i_1 \neq i'_1$ and $i'_2 \neq i_2$.
  Because the cycles have no forbidden pairs,
  $i_1 \neq i_2$ and $i'_1 \neq i'_2$.
  Thus we can have at most two equalities among the four possible comparisons
  between $\{i_1,i_2\}$ and $\{i'_1,i'_2\}$.
  Therefore one of the two merging strategies works.

  We can implement this algorithm in linear time by constructing the pairing
  locally as linked pointers, representing each cycle as a doubly linked list
  on its edges, where each edge stores its two neighboring edges in the cycle
  in no particular order.
  Number the cycles $1, 2, \dots, k$, and iterate over the cycles
  to mark each vertex with each of the cycles it belongs to,
  along with one edge pairing from that cycle.
  Label cycle $1$ as ``merged'' and the rest as ``unmerged''.
  Perform a depth-first search in $G$ from any vertex that is in cycle~$1$.
  At each vertex $v$ visited, iterate through the cycles that $v$ belongs to
  (via $v$'s marks); if any cycle $i$ has not yet been merged, then merge it
  into cycle $1$ by adjusting $O(1)$ pointers among $v$'s marked edge pairings
  for cycles $1$ and~$i$, labeling cycle $i$ as ``merged''.
  By induction, every vertex visited by the depth-first search will have
  already been merged into cycle~$1$.
  The running time beyond the linear cost of depth-first search
  is proportional to the number of marks, which 
  (by the Handshaking Lemma) is twice the number of edges.
  This algorithm is essentially the efficient implementation of
  Hierholzer's Algorithm for Eulerian tours from \cite{Fleischner-1991}.
\end{proof}

Next we combine this result with the results of
Section~\ref{sec:alternating euler} about antidirected Eulerian paths.
For a \emph{directed} graph $G$ and a partition system $P$,
define a \defn{forbidden-transition antidirected Eulerian path} in $(G,P)$
to be an antidirected Eulerian path $e_1, e_2, \dots, e_{|E|}$ of $G$ such that
no two edges $e_i$ and $e_{i+1}$ belong to a common group $P_{v,j}$
where $v$ is the shared vertex of $e_i$ and $e_{i+1}$.

\begin{corollary} \label{cor:forbidden alternating euler}
  The forbidden-transition antidirected Eulerian path problem
  can be solved in linear time.
  The same result holds if the path is further restricted to start and/or
  end with specified direction (forwards or backwards).
\end{corollary}

\begin{proof}
  Apply the reduction of Theorem~\ref{thm:alternating euler}
  to obtain an undirected graph $G'$ with the property that Eulerian paths
  in $G'$ correspond to antidirected Eulerian paths in~$G$.
  For each vertex $v^\pm$ of $G'$ and each $1 \leq i \leq k_v$,
  define $P'_{v^\pm,i}$ to be the set of edges of $G'$ incident to $v^\pm$
  that correspond to edges of $G$ in $P_{v,i}$.
  Then apply Theorem~\ref{thm:forbidden Euler} to decide whether $(G',P')$
  has a forbidden-transition Eulerian path, which is equivalent to
  whether $(G,P)$ has a forbidden-transition antidirected Eulerian path.
  To handle the start/end direction constraints, we can apply the same
  post-analysis as in Corollary~\ref{cor:alternating euler}.
\end{proof}

\subsubsection{Linear-Time Algorithm for Leg-Contact Right-Isosceles-Triangle Edge Matching}
\label{right triangle P}

Now we use the algorithms we have built for antidirected and
forbidden-transition Eulerian paths to solve leg-contact
right-isosceles-triangle edge matching.  The unsigned case reduces to
antidirected Eulerian paths, while the signed case reduces to
forbidden-transition antidirected Eulerian paths.

\begin{theorem}\label{thm:right triangle P}
  $1\times n$ signed and unsigned leg-contact right-isosceles-triangle
  edge-matching puzzles can be solved in linear time.
\end{theorem}

\begin{proof}
First note that tile hypotenuses can never touch in a $1 \times n$ box
by leg contact, so we can ignore those edges' colors completely.
We treat the signed and unsigned cases separately:

\medskip

\textbf{Unsigned case:}
Our algorithm reduces unsigned edge matching to the antidirected Eulerian path
problem in a directed graph, as solved in Section~\ref{sec:alternating euler}.
Given an instance of unsigned $1 \times n$ leg-contact isosceles-right-triangle
edge matching, we construct a directed graph $G$ as follows.
Create a vertex for each unique color that occurs on the legs of the tiles.
For every triangle $\TRITILER{u}{v}{H}$, create a directed edge $(u,v)$.

Any edge-matching solution consists of some ordering of the triangles
that they pack into the $1 \times n$ box, with triangles alternating between
being oriented with its hypotenuse on the top or bottom
(see Figure~\ref{fig:leg}), and consecutive triangles matching on their
shared edges.
We claim that such an edge-matching solution corresponds, by replacing each
tile with its corresponding edge in~$G$, to an antidirected Eulerian path
in~$G$.
First, the path must be antidirected:
following an edge $(u,v)$ in the forwards direction corresponds to placing
$\TRITILER{u}{v}{H}$ with its hypotenuse on the bottom (so colors $u$ and $v$
on the left and right, respectively), while following edge $(u,v)$ in the
reverse direction $(v,u)$ corresponds to placing the tile rotated
$180^\circ$ with its hypotenuse on the top (so colors $v$ and $u$ on the
left and right, respectively).
Second, the path must be Eulerian,
because an edge-matching solution must use every tile exactly once.

The last constraint to handle is the left and right boundary conditions.
If the left edge of the box has an acute angle at the bottom [top], then the
first tile must be placed with its hypotenuse on the bottom [top], so the first
edge of the antidirected Eulerian path must be forwards [backwards].
Similarly, if the right edge of the box has an acute angle at the bottom [top],
then the last tile must be placed with its hypotenuse on the bottom [top],
so the first edge of the antidirected Eulerian path must be forwards
[backwards].
These constraints are exactly what Corollary~\ref{cor:alternating euler}
handles in polynomial time.
By deciding whether $G$ has an appropriate antidirected Eulerian path,
we decide whether the edge-matching puzzle has a solution,
and an actual solution can be converted by the tile--edge correspondence.

\medskip

\textbf{Signed case:}
Our algorithm reduces signed edge matching to the forbidden-transition
antidirected Eulerian path problem in a directed graph,
as solved in Section~\ref{sec:forbidden euler}.
Given an instance of signed $1 \times n$ leg-contact isosceles-right-triangle
edge matching, we construct the same directed graph $G$ as the unsigned case.
To capture the color sign constraint on adjacent tiles,
we define forbidden transitions for the antidirected Eulerian path in~$G$.
Specifically, for each vertex corresponding to an unsigned color~$c$,
define four groups:
\begin{enumerate}
\item $P_{c,1}$ consists of all edges incoming to $c$
  corresponding to tiles of the form $\TRITILER{~}{+c}{}$;
\item $P_{c,2}$ consists of all edges incoming to $c$
  corresponding to tiles of the form $\TRITILER{~}{-c}{}$;
\item $P_{c,3}$ consists of all edges outgoing from $c$
  corresponding to tiles of the form $\TRITILER{+c}{~}{}$; and
\item $P_{c,4}$ consists of all edges outgoing from $c$
  corresponding to tiles of the form $\TRITILER{-c}{~}{}$.
\end{enumerate}

We claim that edge-matching solutions correspond to
forbidden-transition antidirected Eulerian paths in $(G,P)$.
Any antidirected path, when visiting a vertex $c$ not as a path endpoint,
will use either two incoming edges (groups $1$ and~$2$)
or two outgoing edges (groups $3$ and~$4$).
The forbidden transitions thus exactly prevent matching together two instances
of $c$ of the same sign.
Therefore Corollary~\ref{cor:forbidden alternating euler},
with the same start/end conditions as the unsigned case, solves the problem.
\end{proof}

\subsubsection{\#P-completeness of Leg-Contact Right-Isosceles-Triangle Edge Matching}
\label{right triangle count}

Even though leg-contact right-isosceles-triangle edge-matching puzzles
are not hard to solve, counting their solutions remains hard.

\begin{theorem}
  $1\times n$ signed and unsigned leg-contact right-isosceles-triangle
  edge-matching puzzles are \#P-complete.
\end{theorem}

\begin{proof}
  We reduce from counting the number of Eulerian cycles in an undirected graph,
  proved \#P-complete in \cite{countingEuler}.
  Given such an undirected graph $G$, we first add two vertices $s,t$ and
  attach them to an arbitrary vertex $v$ of~$G$,
  forming an undirected graph~$G'$.
  The number of Eulerian cycles in $G$ is exactly twice the number of Eulerian
  paths in $G'$ (whose endpoints are necessarily $s$ and $t$ --- which endpoint
  is the start of the path incurs the factor of~$2$).
  Thus we can reduce from counting the number of Eulerian paths in a graph
  $G'$ with two degree-$1$ vertices $s,t$.

  \medskip

  \textbf{Unsigned case:}
  For the endpoint vertices $s,t$, construct two corresponding triangles
  $$\TRITILER{U_1}{s}{H} \quad \text{and} \quad \TRITILER{U_2}{t}{H},$$
  where $s$ and $t$ are colors representing those vertices,
  $H$ is an arbitrary hypotenuse color, and
  $U_1$ and $U_2$ are globally unique colors.
  Because $U_1$ and $U_2$ appear only in these tiles,
  the tiles must be placed leftmost and rightmost in the puzzle
  (where the rightmost tile is rotated $180^\circ$).

  For each edge $e = \{u,v\}$ in~$G'$, construct two corresponding triangles
  $$\TRITILER{e}{u}{H} \quad \text{and} \quad \TRITILER{e}{v}{H},$$
  where $u,v$ are colors representing these vertices and
  $e$ is a color representing this edge.
  Because color $e$ appears only in these two tiles, these tiles must be placed
  together (with one of them rotated $180^\circ$),
  resulting in a parallelogram with left color $u$ and right color~$v$ or,
  rotating $180^\circ$, the same shape with left color $v$ and right color~$u$.
  Thus these two tiles (or the resulting parallelogram) simulates the
  edge $\{u,v\}$ that can be used in either direction.

  It follows that edge-matching solutions correspond bijectively to
  Eulerian paths in~$G'$.

  \medskip

  \textbf{Signed case:}
  For the endpoint vertices $s,t$, construct two corresponding triangles
  $$\TRITILER{U_1}{+s}{H} \quad \text{and} \quad \TRITILER{U_2}{+t}{H},$$
  where $s$ and $t$ are colors representing those vertices,
  $H$ is an arbitrary hypotenuse color, and
  $U_1$ and $U_2$ are globally unique colors, forcing these tiles
  to be placed leftmost and rightmost in the puzzle.

  For each vertex $v \notin \{s,t\}$ in $G'$, which has even degree~$k$,
  construct $k/2$ copies of two corresponding triangles
  $$\TRITILER{-v_X}{+v}{H} \quad \text{and} \quad \TRITILER{+v_X}{+v}{H},$$
  where $v,v_X$ are two colors corresponding to vertex~$v$.
  Because $v_X$ appears only in these two triangles,
  they must be placed together (with one of them rotated $180^\circ$)
  to match up the $v_X$-color edges,
  resulting in a parallelogram with end colors $+v$ and~$+v$.

  For each edge $e = \{u,v\}$ in~$G'$, construct two corresponding triangles
  $$\TRITILER{-e}{-u}{H} \quad \text{and} \quad \TRITILER{+e}{-v}{H},$$
  where $u,v$ are colors representing these vertices and $e$ is a color
  representing this edge.
  (This construction depends slightly on how we distinguish the endpoints of $e$
  as $u$ and $v$, but the choice can be made arbitrarily for each edge
  without affecting the rest of the construction.)
  Because color $e$ appears only in these two tiles, these tiles must be placed
  together (with one of them rotated $180^\circ$),
  resulting in a parallelogram with left color $-u$ and right color~$-v$ or,
  rotating $180^\circ$, the same shape with left color $-v$ and right color~$-u$.

  By the signs of the colors, any edge-matching solution must alternate
  between edge parallelograms and vertex parallelograms, starting and ending
  with edge parallelograms, surrounded by the $s$ and $t$ triangles.
  It follows that edge-matching solutions correspond to Eulerian paths in~$G'$.

  This reduction is not parsimonious.
  Each vertex parallelogram (with the same external colors of~$+v$)
  can be formed in two ways, blowing up the number of solutions
  by a factor of~$2$.
  If $G'$ has $m$ edges, then there are
  $m-1 = \sum_{v \notin \{s,t\}} \degree(v)/2$ such vertex parallelograms,
  for a total blowup of $2^{m-1}$.
  Furthermore, if we do not treat copies of the vertex tiles as identical,
  then the $k/2$ copies of each degree-$k$ vertex tile can be permuted
  arbitrarily, blowing up the number of solutions by a factor of $(k/2)!^2$.
  The total blowup is thus $c = 2^{m-1} \prod_v (\degree(v)/2)!^2$,
  an easy-to-compute constant, making the reduction $c$-monious.
\end{proof}

\section{Shapeless Edge Matching}
\label{sec:shapeless}

In this section, we analyze the complexity of the following problems:
\begin{definition}
\defn{Signed/unsigned shapeless edge matching} is the following problem:
given a set of $n$ unit square tiles where each edge of each tile is given a color (and a sign in the signed case), can the tiles be laid out in any configuration in the plane such that the overall arrangement is connected via edges, and all edge-to-edge contacts between tiles are compatible?  In the \defn{rooted} variant, the problem specifies a single tile to be fixed at the origin in a specified orientation.
\end{definition}
The distinguishing feature of this problem, compared to the rectangular edge-matching problems for which hardness is already known, is that the target shape is not specified, so there is no constraint on the spatial footprint of a solution. We will show that shapeless edge matching is NP-complete and rooted shapeless edge matching is ASP-complete and \#P-complete, by reduction from $1 \times n$ edge matching with specified left boundary color, which was proved NP-complete by \cite{bosboom2017even} and proved ASP/\#P-complete in Section~\ref{sec:ASP} of this paper (for both the signed and unsigned cases).

\subsection{Shapeless Edge Matching NP-completeness}

\begin{theorem} \label{thm:shapeless}
Signed and unsigned shapeless edge-matching puzzles are NP-complete.
\end{theorem}

\begin{proof}
A shapeless edge-matching solution can clearly be checked in polynomial time,
so shapeless edge matching is in NP.

To prove NP-hardness, we reduce from $1 \times n$ edge matching with specified left boundary color.
Suppose we are given an instance consisting of a set $T$ of $n$ tiles
(signed or unsigned) and a single color $L$ denoting the color of the left
boundary edge of the $1 \times n$ target box.
We will produce a shapeless edge-matching instance consisting of tile set
$T \cup T'$, where $|T'| = O(|T|) = O(n)$.

We design tile set $T'$ to force these tiles into a rectangular frame
structure that simulates a $1 \times n$ box.
Figure~\ref{fig:shapeless tiles} lists the tiles,
and Figure~\ref{fig:forced frame} shows their intended placement.
We use four new colors $\{TW,RW,BW,LW\}$ that appear positively and negatively
(or in the unsigned case, without signs);
each instance of $U$ represents a globally unique (and hence unmatchable) color.

\begin{figure}
	\centering
	\includegraphics[scale=1]{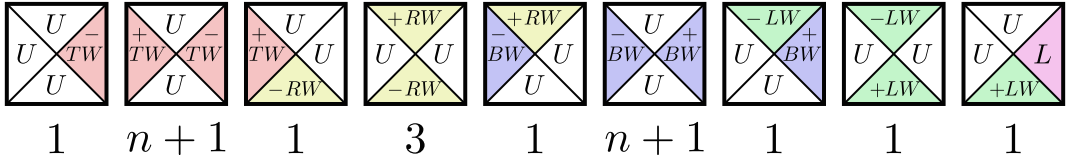}
	\caption{Frame tile set, each labeled by their multiplicity.}
	\label{fig:shapeless tiles}
\end{figure}

\begin{figure}
	\centering
	\includegraphics[scale=1]{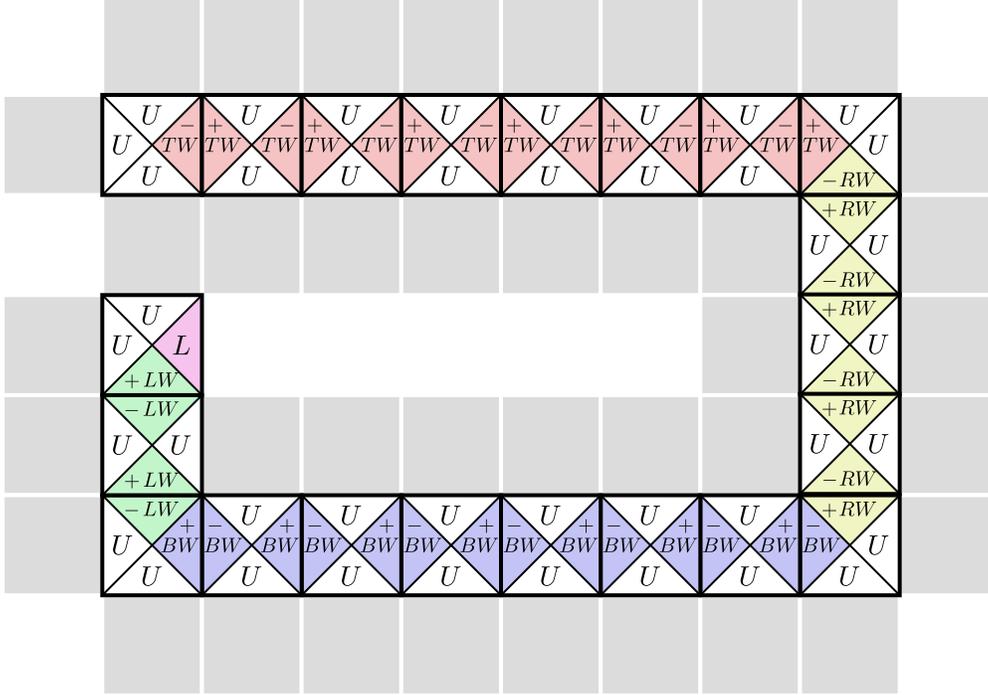}
	\caption{Frame tiles laid out in their forced positions for $n=5$. Grey squares show regions which cannot be occupied by further tiles because they are adjacent to $U$-colored edges.}
	\label{fig:forced frame}
\end{figure}

Next we show that the frame tiles in $T'$ must be positioned to form the frame
shown in Figure~\ref{fig:forced frame}.
Our proof mentions signed tiles, but does not depend on these signs,
so works equally well in the unsigned case by dropping signs from all tiles.
Consider the outer cap \TILE{U}{U}{-TW}{U}. Because the overall arrangement of tiles must be connected but edges colored $U$ are unmatchable, the outer cap's edge colored $TW$ must be adjacent to either a top-wall tile \TILE{+TW}{U}{-TW}{U} or the top-right corner \TILE{+TW}{U}{U}{-RW}, the only other tiles with edges colored $TW$.
If the top-right corner were adjacent to the outer cap, it would be impossible to connect any of the $n+1$ top-wall tiles, as there would be no further way to expose an edge colored $TW$ (of either sign).
By induction, all top-wall tiles are forced to be placed in a row adjacent to the outer cap before the top-right corner is placed, being the only remaining tile with an edge colored $TW$.
By the same argument, the right-wall tiles \TILE{U}{+RW}{U}{-RW} and bottom-right corner \TILE{-BW}{+RW}{U}{U} are the only tiles with edges colored $RW$, thus following the top-right corner must be the three right-wall tiles and then the bottom-right corner, and similarly along the bottom wall and left wall, terminating with the left boundary tile \TILE{U}{U}{L}{+LW} as the final frame tile, forming the frame as desired.

Finally, we show that
the shapeless edge-matching puzzle $T \cup T'$ has a solution if and only if
the corresponding $1 \times n$ edge-matching instance $T$ has a solution.
The forced arrangement of frame tiles only exposes edges colored with an
unmatchable $U$ color, except for the single exposed edge colored~$L$.
Thus the input tiles of $T$ must connect to the frame through that single edge.
Figure~\ref{fig:forced frame} shows that the only available region
in which to arrange the tiles of $T$ is within a $1 \times n$ box
with its leftmost boundary colored~$L$.
\end{proof}

\subsection{Shapeless Edge Matching ASP/\#P-completeness}

\begin{corollary}
Signed and unsigned rooted shapeless edge-matching puzzles are ASP-complete and \#P-complete.
\end{corollary}

\begin{proof}
For ASP/\#P-completeness, we reduce from the rooted variant of shapeless edge matching (which specifies the position and orientation of one tile) to avoid the infinite number of translations as well as global rotations.
We follow a similar reduction as the proof of Theorem~\ref{thm:shapeless},
but modified so that the frame has a unique construction, making the reduction
parsimonious.
By Theorem~\ref{thm:asp}, $1 \times n$ signed/unsigned edge matching with specified left boundary color is ASP/\#P-completeness,
so this parsimonious reduction gives us ASP/\#P-completeness for signed/unsigned shapeless edge matching.

The only degree of freedom in Theorem~\ref{thm:shapeless}'s frame construction
is the ordering of the wall tiles along each wall. In order to fix their order,
we create distinct tiles for each position along the wall, and give them
each unique colors only shared with their neighbors in that ordering.
For example, we modify the upper wall to consist of $n+1$ unique upper-wall
tiles and a suitably modified outer cap and upper-right corner as follows:
$$
\TILE{U}{U}{-TW_1}{U}
\quad
\TILE{+TW_1}{U}{-TW_2}{U}
\quad \cdots \quad
\TILE{+TW_i}{U}{-TW_{i+1}}{U}
\quad \cdots \quad
\TILE{+TW_{n+1}}{U}{-TW_{n+2}}{U}
\quad
\TILE{+TW_{n+2}}{U}{U}{-RW_1}.
$$
Applying the same modification to the other walls and corners gives us a frame
that has a unique construction, and thus the number of solutions to the
shapeless edge-matching instance corresponds exactly to the number of solutions
to the original $1 \times n$ edge-matching puzzle with specified left boundary.
\end{proof}

\section{2-player $1 \times n$ Edge Matching}
\label{sec:2-player}

In this section, we prove PSPACE-hardness for 2-player variants of
$1 \times n$ edge matching.
In Section~\ref{sec:par}, we introduce and analyze a new variant of geography
called partizan geography.
Then in Section~\ref{sec:2-player reductions}, we reduce from geography
and our new variant to 2-player $1 \times n$ edge matching.

\subsection{Partizan Geography}
\label{sec:par}

\defn{Geography} (also called generalized geography) is a game played on a
directed or undirected graph with a designated start vertex.
In \defn{vertex geography} \cite{lichtenstein1980go,fraenkel1993undirected},
players take turns moving from the current vertex
to a neighboring vertex that has not been visited,
with the player who cannot move losing.
In \defn{edge geography}
\cite{DBLP:journals/jcss/Schaefer78,fraenkel1993undirected},
revisiting vertices is allowed, but each edge can be used only once.
In all four variants, directed/undirected vertex/edge geography,
the decision question is whether the first player has a winning strategy.
Undirected vertex geography can be solved in polynomial time
\cite{fraenkel1993undirected}, while all three other versions are
PSPACE-complete \cite{lichtenstein1980go,DBLP:journals/jcss/Schaefer78,fraenkel1993undirected}.

We introduce partizan versions of geography, where the available moves
depend on which player is moving next.
In \defn{$X\!$ $Y\!$-partizan $Z\!$ geography}, with
$X \in \{\text{directed}, \text{undirected}\}$ and
$Y,Z \in \{\text{vertex}, \text{edge}\}$,
players take turns in an $X$ graph extending a shared path,%
\footnote{Fraenkel and Simonson~\cite{DBLP:journals/tcs/FraenkelS93} analyze
  ``path-construction games'' with two paths, with partizan and
  impartial variants that specify which paths each player is allowed to extend.
  Tron~\cite{DBLP:conf/fun/Miltzow12} is another PSPACE-complete two-player
  two-path game.  By contrast, partizan geography is about two players
  building a single path (like geography).}
playing only $Y$s of their color while not repeating any $Z$ already visited.
For example, in edge-partizan vertex geography, players can play only edges of
their color that lead to a vertex not already visited.
We give a complete characterization for $X$ $Y$-partizan $Z$ geography for all
combinations of $X,Y,Z$, as summarized in Table~\ref{tab:partizan-geography}.

First we need a result about (impartial) geography that has been widely
assumed, but to the best of our knowledge, not explicitly proved in the
literature:

\begin{theorem} \label{thm:bipartite-directed-edge-geography}
Directed edge geography remains PSPACE-hard even when restricted to bipartite planar graphs with maximum degree 3 and maximum in/outdegree 2.
\end{theorem}

Problem GP2 in Garey and Johnson~\cite{Garey-Johnson-1979} is called
simply ``Generalized Geography'', but its decision question describes
directed edge geography, and they cite Schaefer's paper
\cite{DBLP:journals/jcss/Schaefer78} which gives a PSPACE-hardness proof.
But Garey and Johnson also cite Lichtenstein and
Sipser~\cite{lichtenstein1980go} to add the bipartite, planar, and degree restrictions on the graph, apparently overlooking that the latter paper is about vertex geography.
This claim and citation pair have been repeated, including in Fraenkel et al.'s paper on undirected geography~\cite{fraenkel1993undirected}, though Bodlaender~\cite{DBLP:journals/tcs/Bodlaender93} correctly distinguishes.

\begin{proof}
Directed vertex geography is PSPACE-hard on bipartite planar graphs with maximum degree 3 and maximum in/outdegree 2~\cite{lichtenstein1980go}.  We reduce from vertex to edge geography by replacing each vertex (with any number of incoming and outgoing edges) with the gadget shown in Figure~\ref{fig:bipartite-vertex-to-edge}.  This gadget is bipartite and planar, and it has the same maximum indegree and outdegree as the vertex it replaces.

\begin{figure}
$$
\vcenter{\hbox{\includegraphics[scale=1.5]{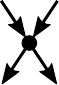}}}
\qquad\vcenter{\hbox{\scalebox{2}{$\to$}}}\qquad
\vcenter{\hbox{\includegraphics[scale=1.5]{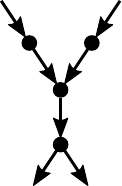}}}
$$
\caption{Gadget simulating vertex geography in edge geography}
\label{fig:bipartite-vertex-to-edge}
\end{figure}

If player 1 plays any of the incoming edges to this gadget, the next two moves are forced; then it is player 2's turn to play one of the outgoing edges.  Once the gadget has been traversed, playing any of the remaining incoming edges loses (because the central edge has already been played).  Thus this gadget correctly simulates a vertex in the vertex geography instance.
\end{proof}

\begin{theorem} \label{thm:vertex-partizan}
Vertex-partizan geography is equivalent to geography in bipartite graphs.  Specifically:
\begin{itemize}
\item Directed vertex-partizan vertex geography and directed vertex-partizan edge geography are PSPACE-complete even when restricted to bipartite planar graphs with maximum degree 3 and maximum in/outdegree 2.
\item Undirected vertex-partizan vertex geography and undirected vertex-partizan edge geography can be solved in polynomial time.
\end{itemize}
\end{theorem}

\begin{proof}
Given a bipartite geography instance, coloring the vertices according to the bipartition produces a vertex-partizan game with the same winner.  Conversely, no monochromatic edges in a vertex-partizan instance can be played because the players alternate moves, so those edges can be deleted without changing the winner.  The resulting graph is bipartite, with each partition containing only vertices of a single player's color.  Thus the problems are equivalent.

Directed vertex geography in bipartite planar maximum-degree-3 maximum-in/outdegree-2 graphs is proved PSPACE-complete in \cite{lichtenstein1980go} and Theorem~\ref{thm:bipartite-directed-edge-geography} extends this to directed edge geography in the same class of graphs.  Undirected vertex geography (in all graphs) and bipartite undirected edge geography are both polynomial~\cite{fraenkel1993undirected}.  All of these results carry over directly to vertex-partizan geography.
\end{proof}

\begin{theorem} \label{thm:edge partizan}
Edge-partizan geography (of all kinds) is PSPACE-complete even when restricted to bipartite planar graphs with maximum degree 3 and maximum in/outdegree 2.
\end{theorem}

\begin{proof}
Given an (impartial) bipartite directed vertex/edge geography instance, we can color the vertices red and blue, so (by bipartiteness) every edge is from red to blue or from blue to red.  Color the first type of edge red and the second type of edge blue.  Because every path alternates vertex colors, every path also alternates edge colors, so adding the edge-partizan constraint does not prohibit any paths.  Thus bipartite directed geography reduces to directed edge-partizan geography.

We can reduce directed edge-partizan geography to undirected edge-partizan geography using the directed-edge-simulation gadget in Figure~\ref{fig:directed-to-undirected-edge-partizan}.  When the blue player plays the left edge, the red and blue player's next moves are forced; then it is the red player's turn at the right vertex.  If blue tries to play the simulated edge backwards (starting at the right vertex), then red can immediately win using the leaf.

\begin{figure}
  $$
  \vcenter{\hbox{\includegraphics[scale=1.5]{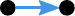}}}
  \qquad\vcenter{\hbox{\scalebox{2}{$\to$}}}\qquad
  \raisebox{-2pt}{\includegraphics[scale=1.5]{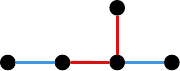}}
  $$
  \caption{A gadget simulating a directed edge with undirected edges. (Exchange colors to simulate a red edge.)}
  \label{fig:directed-to-undirected-edge-partizan}
\end{figure}

Thus all edge-partizan geography games are PSPACE-complete even when restricted to bipartite planar graphs with maximum degree 3 and maximum in/outdegree 2, again carrying through the results in \cite{lichtenstein1980go} and Theorem~\ref{thm:bipartite-directed-edge-geography}.
\end{proof}

\subsection{Reduction from Geography to 2-player $1\times n$ Edge Matching}
\label{sec:2-player reductions}

In this section, we analyze the complexity all four variants of the following
2-player edge-matching game:

\begin{definition}
In the \defn{2-player signed/unsigned edge-matching game}, two players play on
a $1 \times n$ board where the left boundary edge has a specified
(possibly signed) color.  Also given are $n$ square tiles, where each tile
$T_i = \TILE{a_i}{b_i}{c_i}{d_i}$ consists of four (possibly signed) edge colors.
In two variants, the players draw from a \defn{shared pool}
(any player can choose any tile) or from their \defn{own pools}
(each player can choose a tile only from their own pool).
The players take turns making the following type of move:
choose an unused tile from the available pool,
choose one of the four rotations of the tile, and
place the rotated tile in the leftmost unoccupied position of the board.
A move is valid only if the tile's left edge is compatible with the edge to
its left (on the right of the previously played tile or the edge of the board).
If a player has no valid move, then that player loses and the other player wins.
The decision problem is to determine whether the first player can force a win.
\end{definition}

First we present a proof similar to the proof of Theorem~\ref{thm:asp},
although its results are subsumed by the following theorem.

\begin{theorem}
If players draw from a shared pool of tiles, which can be signed or unsigned, the 2-player edge-matching game is PSPACE-complete.
\end{theorem}

\begin{proof}
We reduce from directed vertex geography in graphs with maximum degree 3, which was proved PSPACE-hard in~\cite{lichtenstein1980go}.
Our reduction is the same as the reduction used in the proof of 1-player ASP-completeness in Theorem~\ref{thm:asp}, whose tiles are shown in Figure~\ref{fig:asp}. In the proof of Theorem~\ref{thm:asp}, three tiles are placed for each vertex, so if two players alternate placing tiles, then they alternate placing the first tile for each vertex, which corresponds to taking that vertex in the geography game. In the same proof, the only choices are which tile to place second for each vertex of outdegree 2 (the first tile is fixed, and the unchosen tile must be placed third), a choice which the player who did not place the first tile for that vertex can make and which determines the next visited vertex in the tile-placing game. Correspondingly, in the geography game, when one player chooses a vertex, the player who did not choose that vertex chooses the next visited vertex. Finally, the winner of the tile-placing game is the last player to place a tile. Each vertex has three tiles which are always placed in sequence, so the last player to place a tile is the last player to place the first tile for a vertex, which corresponds to the last player to pick a vertex in the geography game. So the winner of the geography game is the winner of the tile-placing game, as desired.
\end{proof}

The same proof almost works in the case where the players draw from their own
pools of tiles if we reduce from directed \emph{vertex-partizan} vertex
geography, because then we know which player places the first tile
for each vertex.
However, the other player needs to be able to choose the second tile for each
vertex, and then the original player needs to be able to choose the remaining
third tile, meaning we do not know which pools should have those two tiles.
In fact, there is an even simpler proof that avoids this problem:

\begin{theorem}
  The 2-player signed and unsigned edge-matching games are PSPACE-complete,
  whether players draw from their own pools of tiles or from a shared pool.
\end{theorem}

\begin{proof}
  We reduce from a version of edge geography.
  For signed edge matching, we reduce from directed edge geography.
  For unsigned edge matching, we reduce from undirected edge geography.
  For players drawing from their own tile pools,
  we reduce from edge-partizan edge geography.
  For players drawing from a shared pool,
  we reduce from impartial (nonpartizan) edge geography.
  All four of these versions of edge geography are PSPACE-complete
  by \cite{DBLP:journals/jcss/Schaefer78,fraenkel1993undirected}
  and Theorem~\ref{thm:edge partizan}.

  In all cases, the reduction creates a single tile for each edge in the graph.
  For a directed edge $(u,v)$, we make a signed tile $\TILE{-u}{U}{+v}{U}$.
  For an undirected edge $\{u,v\}$, we make an unsigned tile
  $\TILE{u}{U}{v}{U}$.  Each $U$ denotes a globally unique color,
  so these tiles can be rotated only by $180^\circ$.
  In the own-pool case, we put the tile in the pool of the player
  that can play the corresponding edge in edge-partizan geography.
  We set the left boundary edge color to $+s$ in the signed case
  and $s$ in the unsigned case, where $s$ is the given start vertex.
  We define the board size $n$ to be the number of tiles
  (the number of edges in the input graph) so that there
  is no additional limit on the number of moves.

  We claim that the resulting 2-player edge matching game faithfully simulates
  the edge geography game.  By the left edge color, the first tile must have
  an edge colored~$s$, and in the signed case, the edge must be colored~$-s$;
  equivalently, the first edge played in geography must be incident to~$s$,
  and in the directed case, it must be an edge outgoing from~$s$.
  In a general move, the rightmost tile's right edge (exposed) color $v$
  represents the vertex $v$ most recently visited by the path, and the current
  player must choose a tile representing an edge incident to or outgoing from
  that vertex, revealing the other endpoint of that edge.
  Because each tile can be played only once, each edge can be played only once
  (edge geography).  The last player to play a tile/edge wins.
\end{proof}

\section*{Acknowledgments}

This work was initiated during open problem solving in the MIT class on
Algorithmic Lower Bounds: Fun with Hardness Proofs (6.892) in Spring 2019.
We thank the other participants of that class for related discussions and
providing an inspiring atmosphere.

\bibliographystyle{alpha}
\bibliography{biblio}

\newcommand{\etalchar}[1]{$^{#1}$}
\begin{thebibliography}{BJBJK17}

\bibitem[AGW19]{alternating-euler-independent}
David Ariyibi, Jonathan Gabor, and Aaron Williams.
\newblock Personal communication, 2019.

\bibitem[BDD{\etalchar{+}}17]{bosboom2017even}
Jeffrey Bosboom, Erik~D. Demaine, Martin~L. Demaine, Adam Hesterberg, Pasin
  Manurangsi, and Anak Yodpinyanee.
\newblock Even $1 \times n$ edge-matching and jigsaw puzzles are really hard.
\newblock {\em Journal of Information Processing}, 25:682--694, 2017.

\bibitem[Ber66]{Berger-1966}
Robert Berger.
\newblock The undecidability of the domino problem.
\newblock {\em Memoirs of the American Mathematical Society}, 66, 1966.

\bibitem[Ber78]{Berman-1978}
Kenneth~A. Berman.
\newblock Aneulerian digraphs and the determination of those {E}ulerian
  digraphs having an odd number of directed {E}ulerian paths.
\newblock {\em Discrete Mathematics}, 22(1):75--80, 1978.

\bibitem[BJBJK17]{antistrong}
J{\aa}rgen Bang-Jensen, St{\'e}phane Bessy, Bill Jackson, and Matthias
  Kriesell.
\newblock Antistrong digraphs.
\newblock {\em Journal of Combinatorial Theory, Series B}, 122:68--90, 2017.

\bibitem[BM76]{Bondy-Murty-1976}
J.~A. Bondy and U.~S.~R. Murty.
\newblock {\em Graph Theory with Applications}.
\newblock North-Holland, 1976.

\bibitem[Bod93]{DBLP:journals/tcs/Bodlaender93}
Hans~L. Bodlaender.
\newblock Complexity of path-forming games.
\newblock {\em Theor. Comput. Sci.}, 110(1):215--245, 1993.

\bibitem[BW05]{countingEuler}
Graham~R. Brightwell and Peter Winkler.
\newblock Counting {E}ulerian circuits is {\#}{P}-complete.
\newblock In {\em Proceedings of the 7th Workshop on Algorithm Engineering and
  Experiments and the 2nd Workshop on Analytic Algorithmics and Combinatorics},
  pages 259--262, Vancouver, Canada, 2005.

\bibitem[DD07]{demaine2007jigsaw}
Erik~D. Demaine and Martin~L. Demaine.
\newblock Jigsaw puzzles, edge matching, and polyomino packing: Connections and
  complexity.
\newblock {\em Graphs and Combinatorics}, 23(1):195--208, 2007.

\bibitem[Dem19]{6.892-L10}
Erik~D. Demaine.
\newblock Lecture 10: {\#P and ASP}.
\newblock In {\em MIT class 6.892: Algorithmic Lower Bounds: Fun with Hardness
  Proofs}. 2019.
\newblock \url{http://courses.csail.mit.edu/6.892/spring19/lectures/L10.html}.
  Originally
  \href{http://courses.csail.mit.edu/6.890/fall14/lectures/L15.html}{Lecture 15
  of MIT class 6.890 in 2014}.

\bibitem[DFZ11]{ding2011minimum}
Liang Ding, Bin Fu, and Binhai Zhu.
\newblock Minimum interval cover and its application to genome sequencing.
\newblock In Weifan Wang, Xuding Zhu, and Ding-Zhu Du, editors, {\em
  Proceedings of the 5th International Conference on Combinatorial Optimization
  and Applications}, pages 287--298, Zhangjiajie, China, 2011.

\bibitem[Fil19]{ivan-thesis}
Ivan Tadeu Ferreira~Antunes Filho.
\newblock Characterizing boolean satisfiability variants.
\newblock M.eng.~thesis, Massachusetts Institute of Technology, 2019.

\bibitem[Fle90]{Fleischner-1990}
Herbert Fleischner.
\newblock Chapter {VI}: Various types of {E}ulerian trails.
\newblock In {\em Eulerian Graphs and Related Topics: Part 1, Volume 1},
  volume~45 of {\em Annals of Discrete Mathematics}. North-Holland, 1990.

\bibitem[Fle91]{Fleischner-1991}
Herbert Fleischner.
\newblock Chapter {X}: Algorithms for {E}ulerian trails and cycle
  decompositions, maze search algorithms.
\newblock In {\em Eulerian Graphs and Related Topics: Part 1, Volume 2},
  volume~50 of {\em Annals of Discrete Mathematics}. North-Holland, 1991.

\bibitem[FS93]{DBLP:journals/tcs/FraenkelS93}
Aviezri~S. Fraenkel and Shai Simonson.
\newblock Geography.
\newblock {\em Theoretical Computer Science}, 110(1):197--214, 1993.

\bibitem[FSU93]{fraenkel1993undirected}
Aviezri~S Fraenkel, Edward~R Scheinerman, and Daniel Ullman.
\newblock Undirected edge geography.
\newblock {\em Theoretical Computer Science}, 112(2):371--381, 1993.

\bibitem[GJ79]{Garey-Johnson-1979}
Michael~R. Garey and David~S. Johnson.
\newblock {\em Computers and Intractability: {A} Guide to the Theory of
  {NP}-Completeness}.
\newblock W. H. Freeman \& Co., 1979.

\bibitem[Gr{\"{u}}71]{Grunbaum-1971}
Branko Gr{\"{u}}nbaum.
\newblock Antidirected hamiltonian paths in tournaments.
\newblock {\em Journal of Combinatorial Theory, Series B}, 11(3):249--257,
  1971.

\bibitem[IMRS98]{DBLP:journals/siamcomp/HuntMRS98}
Harry B.~Hunt III, Madhav~V. Marathe, Venkatesh Radhakrishnan, and
  Richard~Edwin Stearns.
\newblock The complexity of planar counting problems.
\newblock {\em {SIAM} Journal on Computing}, 27(4):1142--1167, 1998.

\bibitem[Kot68]{Kotzig1968}
Anton Kotzig.
\newblock Moves without forbidden transitions in a graph.
\newblock {\em Matematick\'y \v{c}asopis}, 18(1):76--80, 1968.

\bibitem[LS80]{lichtenstein1980go}
David Lichtenstein and Michael Sipser.
\newblock Go is polynomial-space hard.
\newblock {\em Journal of the ACM}, 27(2):393--401, 1980.

\bibitem[Mil12]{DBLP:conf/fun/Miltzow12}
Tillmann Miltzow.
\newblock Tron, a combinatorial game on abstract graphs.
\newblock In Evangelos Kranakis, Danny Krizanc, and Flaminia~L. Luccio,
  editors, {\em Proceedings of the 6th International Conference on Fun with
  Algorithms}, volume 7288 of {\em Lecture Notes in Computer Science}, pages
  293--304. Springer, June 2012.

\bibitem[Ple79]{plesnik}
J{\'a}n Plesn{\'i}k.
\newblock The {NP}-completeness of the {H}amiltonian cycle problem in planar
  diagraphs with degree bound two.
\newblock {\em Information Processing Letters}, 8(4):199--201, April 1979.

\bibitem[Sch78]{DBLP:journals/jcss/Schaefer78}
Thomas~J. Schaefer.
\newblock On the complexity of some two-person perfect-information games.
\newblock {\em Journal of Computer and System Sciences}, 16(2):185--225, 1978.

\bibitem[Set02]{Seta02thecomplexities}
Takahiro Seta.
\newblock The complexities of puzzles, {C}ross {S}um, and their {A}nother
  {S}olution {P}roblems ({ASP}).
\newblock Senior thesis, University of Tokyo, 2002.

\bibitem[Thu92]{Thurston-patent}
E.~L. Thurston.
\newblock Puzzle.
\newblock US Patent 487,798, December 13, 1892.
\newblock \url{https://patents.google.com/patent/US487798}.

\bibitem[Val79]{Valiant-1979-perm}
L.~G. Valiant.
\newblock The complexity of computing the permanent.
\newblock {\em Theoretical Computer Science}, 8(2):189--201, 1979.

\bibitem[Wes01]{intro-to-graphs}
Douglas~B. West.
\newblock {\em Introduction to Graph Theory}.
\newblock Pearson Education, 2nd edition, 2001.

\bibitem[Wik19]{eternity2-wiki}
Wikipedia.
\newblock Eternity {II}.
\newblock \url{https://en.wikipedia.org/wiki/Eternity_II_puzzle}, 2019.

\bibitem[YS03]{Yato-Seta-2003}
Takayuki Yato and Takahiro Seta.
\newblock Complexity and completeness of finding another solution and its
  application to puzzles.
\newblock {\em IEICE Transactions on Fundamentals of Electronics,
  Communications, and Computer Sciences}, E86-A(5):1052--1060, 2003.
\newblock Also IPSJ SIG Notes 2002-AL-87-2, 2002.

\bibitem[{\v{Z}}it96]{Zitnik-1996}
Arjana {\v{Z}}itnik.
\newblock Anti-directed walks in 4-valent graphs.
\newblock Preprint series, volume 34, number 530, University of Ljubljana,
  Ljubljana, Slovenia, 1996.

\end{thebibliography}

\end{document}